\documentclass{article}

\usepackage[reqno]{amsmath}
\usepackage[inference]{semantic}
\usepackage{verbatim}
\usepackage{stmaryrd}
\usepackage{amssymb} 
\usepackage{amsthm} 
\usepackage{bbm}
\usepackage{pifont}
\usepackage{txfonts}
\usepackage[bookmarks]{hyperref}
\usepackage{tikz}
\usepackage{fullpage}

\input xy
\xyoption{all}

\newtheorem{theorem}{Theorem}
\newtheorem{proposition}[theorem]{Proposition}
\newtheorem{lemma}[theorem]{Lemma}

\newtheorem{definition}[theorem]{Definition}

\newtheorem{problem}[theorem]{Problem}

\title{The Name-Passing Calculus}
\author{Yuxi Fu \ \ \  Han Zhu \\
BASICS, Shanghai Jiao Tong University}
\date{}

\begin{document}

\maketitle

\begin{abstract}
Name-passing calculi are foundational models for mobile computing.
Research into these models has produced a wealth of results ranging from relative expressiveness to programming pragmatics.
The diversity of these results call for clarification and reorganization.
This paper applies a model independent approach to the study of the name-passing calculi, leading to a uniform treatment and simplification.
The technical tools and the results presented in the paper form the foundation for a theory of name-passing calculus.
\end{abstract}

\section{Mobility in Practice and in Theory}\label{sec-Mobile-Calculi-in-Practice-and-in-Theory}

Mobile calculi feature the ability to pass around objects that
contain channel names. Higher order
CCS~\cite{Thomsen1989,Thomsen1990,Thomsen1993,Thomsen1995} for
instance, is a calculus with a certain degree of mobility. In a
mobile calculus, a process that receives an object may well make use
of the names which appear in the object to engage in further
interactions. It is in this sense that the communication topology is
dynamic. It was soon realized that the communication mechanism that
restricts the contents of communications to the channel names gives
rise to a simple yet versatile model that is more powerful than the
process-passing
calculi~\cite{Sangiorgi1992,Sangiorgi1993,Sangiorgi1996TCS,Sangiorgi1996IC}.
This is the $\pi$-calculus of Milner, Parrow and
Walker~\cite{MilnerParrowWalker1992}. See~\cite{Parrow2001} for a
gentle introduction to the model and the history of the name-passing
calculus and~\cite{SangiorgiWalker2001BOOK} for a broader coverage.
A seemingly innocent design decision of the $\pi$-calculus is to
admit a uniform treatment of the names. This decision is however not
supported by the semantics of the mobile calculi. From a process
term $T$ one could construct the input prefix term
\begin{equation}\label{prefix-name}
a(x).T
\end{equation}
and the localization term
\begin{equation}\label{localization-name}
(x)T.
\end{equation}
According to the definition of the $\pi$-calculus, the semantics of $x$
which appears in (\ref{prefix-name}) is far different from that of $x$ in
(\ref{localization-name}). In the former $x$ is a name variable, or
a dummy name, that can be instantiated by an arbitrary name when the
prefix engages in an interaction. In the latter $x$ is a local name
that can never be confused with another name. The input prefix
forces the unbound name $x$ in $T$ to be a name variable, whereas
the localization operator forces the unbound name $x$ in $T$ to be a
constant name. This apparent contradiction is behind all the
semantic complications of the $\pi$-calculus. And nothing has been
gained by these complications. In what follows we take a look at some of the issues caused by the confusion.

To begin with, the standard operational semantics of the $\pi$-calculus has not been very smooth.
An extremely useful command in both practice and theory is the two leg if-statement
$\textit{if}\;\varphi\;\textit{then}\;S\;\textit{else}\;T$. In
mobile calculi this can be defined by introducing the conditional
terms $[x{=}y]T$ and $[x{\not=}y]T$. The semantics of these terms
have been defined respectively by the match rule
\begin{equation}\label{good-match}
\inference{T\stackrel{\lambda}{\longrightarrow}T'}{[x{=}x]T\stackrel{\lambda}{\longrightarrow}T'}
\end{equation}
and the mismatch rule
\begin{equation}\label{questionable-mismatch}
\inference{T\stackrel{\lambda}{\longrightarrow}T'}{[x{\not=}y]T\stackrel{\lambda}{\longrightarrow}T'}
\ x\not=y.
\end{equation}
Rule (\ref{questionable-mismatch}) is unusual since it has an
unusual side condition. How should we understand the side condition
$x\not=y$? If $x,y$ were constant names, the side condition of
(\ref{questionable-mismatch}) would be pointless because $x\not=y$
would be evaluated to logical truth. The reason that
(\ref{questionable-mismatch}) is necessary is precisely because
$x,y$ cannot be understood as constant names in the uniform
treatment of the names. The correct reading of
(\ref{questionable-mismatch}) is that
``$[x{\not=}y]T\stackrel{\lambda}{\longrightarrow}T'$ is admissible
under the logical assumption $x\not=y$''. It should now be clear
that the popular semantics fails to support the following
equivalence
\begin{equation}\label{believe-it-or-not}
T = [x{=}y]T+[x{\not=}y]T.
\end{equation}
To see this, let $=$ be $\sim $, the strong early bisimilarity
of Milner, Parrow and Walker~\cite{MilnerParrowWalker1992}. According to the definition, $T
\sim [x{=}y]T+[x{\not=}y]T$ if and only if $T\sigma \dot{\sim }
([x{=}y]T+[x{\not=}y]T)\sigma$ for every substitution $\sigma$,
where $\dot{\sim }$ is the strong ground bisimulation equivalence.
If $\sigma$ is the identity substitution, it boils down to
establishing the following equivalence
\begin{equation}\label{believe-it-or-not-g}
T \;\dot{\sim }\; [x{=}y]T+[x{\not=}y]T.
\end{equation}
We may prove (\ref{believe-it-or-not-g}) under the assumption
$x\not=y$. But we cannot prove (\ref{believe-it-or-not-g}) under the
assumption $x=y$ since rule (\ref{good-match}) does not allow us to do that.
A related mistake is to introduce a boolean evaluation function $beval(\_)$ whose inductive definition includes the following clauses:
\begin{eqnarray*}
beval(x{=}y) &\stackrel{\rm def}{=}& \bot, \\
beval(x{\not=}y) &\stackrel{\rm def}{=}& \top.
\end{eqnarray*}
This would lead to the axioms
\begin{eqnarray*}
\;[x{=}y]T &=& {\bf 0}, \\
\;[x{\not=}y]T &=& T,
\end{eqnarray*}
which are wrong if observational equivalences are closed under prefix operations.
One way to rectify (\ref{good-match}) is to introduce the following
\begin{equation}\label{ok-match}
\inference{T\stackrel{\lambda}{\longrightarrow}T'}{[x{=}y]T\stackrel{\lambda}{\longrightarrow}T'}\ x=y.
\end{equation}
The rule (\ref{ok-match}) does not completely eradicate the problem.
Take for example the following instance of the expansion law
\begin{equation}\label{el}
\overline{x}x\,|\,y(v) = \overline{x}x.y(v) + y(v).\overline{x}x +
[x{=}y]\tau.
\end{equation}
The right hand side of (\ref{el}) can do a $\tau$-action under the
assumption $x=y$. However the operational semantics of the $\pi$-calculus does not admit a $\tau$-action of the left hand side
of (\ref{el}) even if the logical assumption $x=y$ is made.
One solution to the above problem is to introduce the following rule
\[
\inference{T\{y/x\}\stackrel{\lambda\{y/x\}}{\longrightarrow}T'\{y/x\}}{T\stackrel{\lambda}{\longrightarrow}T'}\ \mathrm{if}\ x=y\ \mathrm{and}\ bn(\lambda)\cap\{x,y\}=\emptyset.
\]
This rule is necessary because the terms $T\{y/x\},T$ are distinct in the meta logic.
Our scrutiny of the match/mismatch semantics takes us to the symbolic approach of Hennessy and Lin~\cite{HennessyLin1995,Lin1995,Lin1995-fix-ind,Lin1996,Lin1998,Lin2003}.
In the symbolic semantics, one has that
$[x{\not=}y]T\stackrel{x\not=y,\lambda}{\longrightarrow}T'$, meaning
that the action is admissible under the logical assumption
$x\not=y$. Similarly one has
$[x{=}y]T\stackrel{x=y,\lambda}{\longrightarrow}T'$. Notice that
this transition is very different from the transition
$[x{=}x]T\stackrel{\top,\lambda}{\longrightarrow}T'$. In the
symbolic approach the action
$T\stackrel{\top,\lambda}{\longrightarrow}T'$ is simulated by the
combined effect of
$[x{\not=}y]T\stackrel{x\not=y,\lambda}{\longrightarrow}T'$ and
$[x{=}y]T\stackrel{x=y,\lambda}{\longrightarrow}T'$, not by any
single action sequence of $[x{=}y]T+[x{\not=}y]T$.
If we think of it, the symbolic semantics not only provides the correct labeled
transition semantics upon which we may study the observation theory
of the $\pi$-calculus, but also makes clear the problem caused by
the confusion of the name variables and the names.

Secondly the observational theory of mobile processes is made
more complex than it is. One well-known phenomenon is that, unlike
in CCS~\cite{Milner1989,MilnerSangiorgi1992}, some standard
definitions of process equivalence lead to different equality
relations. The standard definition of bisimulation gives rise to
ground bisimilarity that is not closed under input prefix operation.
The solution proposed in~\cite{MilnerParrowWalker1992} is to take
the substitution closure. The resulting relation is the early
equivalence. If substitution closure is required in every
bisimulation step, one obtains Sangiorgi's open
bisimilarity~\cite{Sangiorgi1996AI}. The open bisimilarity is
strictly finer than the early equivalence, which is in turn much
finer than the ground bisimilarity. The open bisimilarity can be
further improved to quasi open
bisimilarity~\cite{SangiorgiWalker2001}, which lies nontrivially
between the open bisimilarity and the early equivalence. The barbed
equivalence can be defined by placing substitution closure at the
beginning of bisimulations, which gives rise to the equivalence
studied by Milner and Sangiorgi~\cite{MilnerSangiorgi1992}. It can also be defined by
requiring that every bisimulation step should be closed under
substitutions of names. It is shown by Sangiorgi and Walker~\cite{SangiorgiWalker2001}
that the latter coincides with the quasi open bisimilarity. It is
easy to see that the barbed equivalence is weaker than the early
equivalence. It is not yet clear however if it is subsumed by the
early equivalence. Putting aside the issue of which of these
equivalences is more authoritative than the rest, we would like to
point out that the substitution closure requirement is an algebraic
requirement rather than an observational one. From the true
spirit of the observation theory, an environment can never detect
any difference between $a(x).\overline{b}c+\overline{b}c.a(x)$ and
$a(x)\,|\,\overline{b}c$, since it can never force the distinct
names $a,b$ to be equal. This issue of reconciling the inconsistency
between the observational view and the algebraic view must be
addressed to achieve a better theory of the mobile processes.

The algebraic requirement also makes the testing theory of mobile
processes hard to comprehend. In the testing theory developed
by De Nicola and Hennessy~\cite{DeNicolaHennessy1984}, the behaviors of a process are
judged by testers. Two processes are testing equivalent if no
testing can detect any behavioral difference between them. Like the
bisimulation approach, the testing approach fails to give rise to a
reasonable equivalence on the mobile processes. In order to respect
the name uniformity and obtain a useful equivalence at the same
time, the algebraic condition must be imposed.
See~\cite{BorealeDeNicola1995} for more on this issue. In some sense
the substitution closure condition completely defeats the philosophy
of the testing theory.

In retrospect, the confusion of the names and the name variables is
not out of the desire to model mobility, since mobility can be
achieved by using the name variables anyway. If channels have a
physical existence, computations or interactions really should not
manipulate channels. What they are supposed to do is to make use of
the channels for the purpose of interaction. According to this
interpretation, all channel names ought to be constant. To model
mobility, the introduction of a dichotomy between the names and the
name variables is not only an obvious choice, it is the only choice.
The variables are there for mobility.

In theory of expressiveness, the name dichotomy provides a basis for
comparing the relative expressive powers of calculi. The
straightforward translation from CCS to the $\pi$-calculus for instance
is fully abstract if in the $\pi$-calculus a line is drawn between the
names and the name variables. The translation takes the equivalent
CCS processes, say $a\,|\,\overline{b}$ and
$a.\overline{b}+\overline{b}.a$, to the equivalent $\pi$-processes
$a(x)\,|\,\overline{b}(y)$ and
$a(x).\overline{b}(y)+\overline{b}(y).a(x)$. If the names are
treated uniformly, the target model would have a much stronger
process equality than the source model. In such a framework it is
not even clear if a reasonably good fully abstract translation from
CCS to the $\pi$-calculus exists. Other expressiveness results can also
be best interpreted using the name dichotomy. Sangiorgi-Thomsen's
encoding of the higher order CCS in the $\pi$-calculus is another
example. The process variables of the higher order CCS are
translated to the name variables of the $\pi$-calculus, while the
names of the former {\em are} the names of the latter. This encoding
is shown to be fully abstract
by Sangiorgi~\cite{Sangiorgi1992,Sangiorgi1993}. Again if the names of the
$\pi$-calculus are treated uniformly, the encoding would not even be
sound. We could give more examples to support the proposition that a
dichotomic understanding should be preferred. But the point is
already made. The names play a universal role in process theory.
Without the assumption that all names are constant, expressiveness
results about process calculi are bound to be
chaotic~\cite{Nestmann2006}.

When applying the mobile calculi to interpret programming phenomena,
the name dichotomy has always been enforced. It is sufficient to
give just one example. An early work was done
by Walker~\cite{Walker1991,Walker1995}, who defined the operational
semantics of an object oriented language in terms of the operational
semantics of the $\pi$-calculus. The idea of the interpretation can
be summarized as follows. An object is modeled by a prefix process
of the form $objn(x).O$, where $objn$ is the name of the object. A
method is interpreted as a replicated process of the form
$!mthd(z).M$, where $mthd$ is the method name. The method can be
invoked by a process of the form $\overline{mthd}(v).P$ that
supplies the value $v$ to the method parameter. Without going into
details, it is already obvious that for this interpretation to make
sense, it is important to maintain a distinction between the names
and the name variables. We could give many other applications of the
mobile calculi. But it suffices to say that in all these
applications, there is a clear cut distinction between the names and
the name variables.

The above discussions lead to the conclusion that, for both
theoretical and practical reasons, the $\pi$-calculus should be
defined using the name dichotomy. The dichotomy has been introduced
in literature using type systems. If one thinks of the type of a
channel name as defining the interface property of the channel, then
the type theoretical solution does not seem appropriate since the
difference between a name and a name variable is not about interface
property. It is our view that the issue should be treated at a more
fundamental level.

This is both a survey paper and research paper.
Since we adopt a new uniform and simplified presentation of the $\pi$-calculus, there are technical contributions throughout the paper.
In view of the huge literature on the $\pi$-calculus~\cite{SangiorgiWalker2001BOOK}, it is a daunting task to give an overall account of the various aspects of the model. Our strategy in this paper is to present the foundational core of the $\pi$-calculus, covering the observational theory, the algebraic theory and the relative expressiveness.
The novelty of our treatment is that, by applying a model independent approach throughout the paper, a great deal of simplification and unification are achieved.
Our effort can be summarized as follows:
\begin{itemize}
\item We show that a concise operational semantics of $\pi$ is available.

\item
We demonstrate that the observational theory of $\pi$ is far less diverse than it has been perceived.

\item We point out that the algebraic theory of $\pi$ is simpler than has been suggested.
\end{itemize}
The above claims are supported by the following technical contributions:
\begin{itemize}
\item A general model independent process equality, the absolute equality, is applied to the $\pi$-calculus.
It is proved that the well known bisimulation equivalences of the
$\pi$-calculus, mentioned in this introduction, either coincide with
a weak version of the absolute equality or can be safely ignored.

\item A model independent equivalence, the box equality, is
defined and applied to the $\pi$-calculus. It is demonstrated that
this new equivalence coincides with the well known rectification of
the testing equivalence in the $\pi$-calculus.

\item Two complete proof systems for the set of the finite
$\pi$-processes are presented, one for the absolute equality, the
other for the box equality.
\end{itemize}

The model independent theory of process calculi is systematically developed in~\cite{FuYuxi}.
In particular the absolute equality and the subbisimilarity used in this paper are taken from~\cite{FuYuxi}.
It should be pointed out however that the present paper has been made self-contained.

Most of the lemmas are stated without proof.
A well-informed reader would have no problem in supplying the proof details.

The rest of the paper is organized into five sections.
Section~\ref{sec-Name-and-Name-Passing-Process} defines our version
of the $\pi$-calculus. Section~\ref{sec-Equality} studies the model
independent observation theory of the $\pi$-calculus.
Section~\ref{sec-Expressiveness} discusses the relative
expressiveness of some well known variants of the $\pi$-calculus.
Section~\ref{sec-Proof-System} presents a uniform account of the
proof systems for the finite $\pi$-processes.
Section~\ref{sec-future-work} points out how a theory of the $\pi$-calculus can be developed using the framework set up in this
paper.

\section{Pi Calculus}\label{sec-Name-and-Name-Passing-Process}

We assume that there is an infinite countable set $\mathcal{N}$ of
{\em names}, an infinite countable set $\mathcal{N}_{v}$ of {\em
name variables}. These sets will be ranged over by different lower
case letters. Throughout the paper the following conventions will be
enforced:
\begin{itemize}
\item The set $\mathcal{N}$ is ranged over by $a,b,c,d,e,f,g,h$.

\item The set $\mathcal{N}_{v}$ is ranged over by $u,v,w,x,y,z$.

\item The set $\mathcal{N}\cup\mathcal{N}_{v}$ is ranged over
by $l,m,n,o,p,q$.
\end{itemize}
A name variable acts as a place holder that need be substantiated by
a name. By its very nature, a name variable cannot be used as a
channel for interaction. Similarly it cannot be used as a message
passed around in a communication.

\subsection{Process}\label{Process}

To give a structural definition of processes, we need to introduce
terms. The set $\mathcal{T}$ of $\pi$-{\em terms} is inductively
generated by the following BNF:
\begin{eqnarray*}
S,T &:=& {\bf 0} \mid \sum_{i\in I}n(x).T_{i} \mid \sum_{i\in I}\overline{n}m_{i}.T_{i} \mid S\,|\,T \mid (c)T \mid [p{=}q]T \mid [p{\not=}q]T \mid \;!\pi.T,
\end{eqnarray*}
where $I$ is a finite nonempty indexing set and
\begin{eqnarray*}
\pi &:=& n(x) \mid \overline{n}m.
\end{eqnarray*}
Here $n(x)$ is an input prefix and $\overline{n}m$ an output prefix.
The {\em nil} process ${\bf 0}$ cannot do anything in any
environment. For each $i\le n$, the component $n(x).T_{i}$ is a {\em
summand} of the {\em input choice} term $\sum_{i\in I}n(x).T_{i}$,
where the name variable $x$ is {\em bound}. A name variable is {\em
free} if it is not bound. Similarly the component
$\overline{n}m_{i}.T_{i}$ is a summand of the {\em output choice}
term $\sum_{i\in I}\overline{n}m_{i}.T_{i}$.
Notice that input and output choices are syntactically simpler than the separated choices~\cite{Palamidessi2003}.
The term $T\,|\,T'$ is a concurrent {\em composition}. The restriction $(c)T$ is in {\em localization}
form, where the name $c$ is {\em local}. A name is {\em global} if
it is not local. The following functions will be used.
\begin{itemize}
\item $gn(\_)$ returns the set of the global names.

\item $ln(\_)$ returns the set of the local names.

\item $n(\_)$ returns the set of the names.

\item $fv(\_)$ returns the set of the free name variables.

\item $bv(\_)$ returns the set of the bound name variables.

\item $v(\_)$ returns the set of the name variables.
\end{itemize}
The guard $[p{=}q]$ is a {\em match} and $[p{\ne}q]$ a {\em
mismatch}. The term $!\pi.T$ is a {\em guarded replication} and
`$!$' a replication operator. The guarded replication is equivalent
to the general replication of the form $!T$. The transformation from
the general replication to the guarded replication makes use of an
auxiliary function $(\_)^{c}$ defined on the replication free terms.
The structural definition is as follows.
\begin{eqnarray*}
({\bf 0})^{c} &\stackrel{\rm def}{=}& {\bf 0}, \\
(\pi.T)^{c} &\stackrel{\rm def}{=}& \pi.(\overline{c}c\,|\,T^{c}), \\
(T_{1}\,|\,T_{2})^{c} &\stackrel{\rm def}{=}& T_{1}^{c}\,|\,T_{2}^{c}, \\
((a)T)^{c} &\stackrel{\rm def}{=}& (a)T^{c}, \\
([p{=}q]T)^{c} &\stackrel{\rm def}{=}& [p{=}q]T^{c}, \\
([p{\ne}q]T)^{c} &\stackrel{\rm def}{=}& [p{\ne}q]T^{c}.
\end{eqnarray*}
If neither $c$ nor $z$ is in $T$, we may define $(!T)^{c}$ by the process $(c)(\overline{c}c\,|\,\overline{c}c\,|\,!c(z).T^{c})$.
It is clear that there would be no loss of expressive power if guarded
replication is further restrained to the form $!p(x).T$ or
$!\overline{p}q.T$.

A {\em finite} $\pi$-term is one that does not contain any
replication operator. A $\pi$-term is {\em open} if it contains free
name variables; it is {\em closed} otherwise. A closed $\pi$-term is
also called a $\pi$-{\em process}. We write $\mathcal{P}$ for the
set of the $\pi$-processes, ranged over by $L,M,N,O,P,Q$. Some
derived prefix operators are defined as follows.
\begin{eqnarray*}
\overline{n}(c).T &\stackrel{\rm def}{=}& (c)\overline{n}c.T, \\
n.T &\stackrel{\rm def}{=}& n(z).T\ \ \mathrm{for}\ \mathrm{some}\
z\notin fv(T), \\
\overline{n}.T &\stackrel{\rm def}{=}& \overline{n}(c).T\ \
\mathrm{for}\ \mathrm{some}\ c\notin gn(T), \\
\tau.T &\stackrel{\rm def}{=}& (c)(c.T\,|\,\overline{c})\ \
\mathrm{for}\ \mathrm{some}\ c\notin gn(T).
\end{eqnarray*}
Furthermore we introduce the following polyadic prefixes:
\begin{eqnarray*}
n(x_{1},\ldots,x_{n}).T &\stackrel{\rm def}{=}&
n(z).z(x_{1}).\cdots.z(x_{n}).T\ \ \mathrm{for}\ \mathrm{some}\
z\notin fv(T), \\
\overline{n}\langle p_{1},\ldots,p_{n}\rangle.T &\stackrel{\rm
def}{=}&
\overline{n}(c).\overline{c}p_{1}.\cdots.\overline{c}p_{n}.T\ \
\mathrm{for}\ \mathrm{some}\ c\notin gn(T),
\end{eqnarray*} where
$n>1$. These two derived operators make it clear how to simulate the
polyadic $\pi$-calculus~\cite{Milner1993-poly} in the (monadic)
$\pi$-calculus.

Both bound name variables and local names are subject to
$\alpha$-conversion. Throughout the paper, it is assumed that
$\alpha$-conversion is applied whenever it is necessary to avoid
confusion. This will be called the $\alpha$-{\em convention}. In for
example the structural composition rule to be defined later, the
side conditions are redundant in the presence of the
$\alpha$-convention.

A {\em condition}, denoted by $\phi,\varphi,\psi$, is a finite
concatenation of matches and/or mismatches. The concatenation of
zero match/mismatch is denoted by $\top$, and its negation is
denoted by $\bot$. We identify syntactically $(\top)T$ with $T$ and
$(\bot)T$ with $\mathbf{0}$. If $\mathcal{F}$ is the finite name set
$\{n_{1},\ldots,n_{i}\}$, the notation
$[p{\not\in}\mathcal{F}]T$ stands for $[p{\not=}n_{1}]\ldots[p{\not=}n_{i}]T$.

A {\em renaming} is a partial injective map
$\alpha:\mathcal{N}\rightharpoonup\mathcal{N}$ whose domain of
definition $dom(\alpha)$ is finite. A {\em substitution} is a partial map
$\sigma:\mathcal{N}_{v}\rightharpoonup\mathcal{N}\cup\mathcal{N}_{v}$
whose domain of definition $dom(\sigma)$ is finite. An {\em assignment} is a
partial map $\rho:\mathcal{N}_{v}\rightharpoonup\mathcal{N}$ that associates a
name to a name variable in the domain of $\rho$. It is convenient to extend the definition
of an assignment $\rho$ by declaring $\rho(a)=a$ for all $a\in\mathcal{N}$. Renaming, substitution and
assignment are used in postfix. We often write
$\{n_{1}/x_{1},\ldots,n_{i}/x_{i}\}$ for a substitution, and similar
notation is used for renaming. The notation $\rho[x\mapsto a]$
stands for the assignment that differs from $\rho$ only in that
$\rho[x\mapsto a]$ always maps $x$ onto $a$ whereas $\rho(x)$ might
be different from $a$ or undefined. Whenever we write $\rho(x)$ we always assume that $\rho$ is defined on $x$.

An assignment $\rho$ satisfies a condition $\psi$, denoted by $\rho
\models \psi$, if $\rho(m) = \rho(n)$ for every $[m{=}n]$ in $\psi$,
and $\rho(m){\ne}\rho(n)$ for every $[m{\ne}n]$ in $\psi$. We write
$\rho\models\psi \Rightarrow \phi$ to mean that $\rho \models \phi$
whenever $\rho \models \psi$, and $\rho\models\psi \Leftrightarrow
\phi$ if both $\rho\models\psi \Rightarrow \phi$ and
$\rho\models\phi \Rightarrow \psi$. We say that $\psi$ is {\em
valid}, notation $\models\psi$ (or simply $\psi$), if
$\rho\models\psi$ for every assignment $\rho$. A useful valid
condition is the following.
\begin{equation}\label{2009-02-02}
(x{\notin}\mathcal{F})\vee \bigvee_{n\in\mathcal{F}}(x{=}n),
\end{equation}
where $\mathcal{F}$ is a finite subset of
$\mathcal{N}\cup\mathcal{N}_{v}$. Given a condition $\varphi$, we
write $\varphi^{=}$ and $\varphi^{\not=}$ respectively for the
condition $\bigwedge\{m{=}n \mid m,n\in n(\varphi)\cup v(\varphi)\
\mathrm{and}\ \models\varphi\Rightarrow m{=}n\}$ and the condition
$\bigwedge\{m{\ne}n \mid m,n\in n(\varphi)\cup v(\varphi)\
\mathrm{and}\ \models\varphi\Rightarrow m{\ne}n\}$.

\subsection{Semantics}\label{Semantics}

The semantics is defined by a labeled transition system structurally
generated by a set of rules. The set $\mathcal{L}$ of {\em labels}
for $\pi$-terms, ranged over by $\ell$, is
\[\{ab,\overline{a}b,\overline{a}(c) \mid a,b,c\in\mathcal{N}\}\]
where $ab,\overline{a}b,\overline{a}(c)$ denote respectively an
input action, an output action and a bound output action. The set
$\mathcal{L}^{*}$ of the finite strings of $\mathcal{L}$ is ranged
over by $\ell^{*}$. The empty string is denoted by $\epsilon$. The
set $\mathcal{A}=\mathcal{L}\cup\{\tau\}$ of actions is ranged over
by $\lambda$ and its decorated versions. The set $\mathcal{A}^{*}$
of the finite strings of $\mathcal{A}$ is ranged over by
$\lambda^{*}$. With the help of the action set, we can define the
operational semantics of the $\pi$-calculus by the following labeled
transition system.

\vspace*{2mm} \noindent{\em Action}
\[\begin{array}{ccc}
\inference{} {\sum_{i\in I}a(x).T_{i}\stackrel{ac}{\longrightarrow
}T_{i}\{c/x\}} \ i\in I \ \ \ & \inference{} {\sum_{i\in I}
\overline{a}c_{i}.T_{i}\stackrel{\overline{a}c_{i}}{\longrightarrow
}T_{i}} \ i\in I
\end{array}\]

\noindent{\em Composition}
\[\begin{array}{ccc}
\inference{T\stackrel{\lambda}{\longrightarrow}T'}{S\,|\,T\stackrel{\lambda}{\longrightarrow}S\,|\,T'}\
\ & \inference{S\stackrel{ab}{\longrightarrow}S'\ \ \ \ \
T\stackrel{\overline{a}b}{\longrightarrow}T'}
{S\,|\,T\stackrel{\tau}{\longrightarrow}S'\,|\,T'} \ \ &
\inference{S\stackrel{ac}{\longrightarrow}S'\ \ \ \ \
T\stackrel{\overline{a}(c)}{\longrightarrow}T'}
{S\,|\,T\stackrel{\tau}{\longrightarrow}(c)(S'\,|\,T')}
\end{array}\]

\noindent{\em Localization}
\[\begin{array}{cc}
\inference{T\stackrel{\overline{a}c}{\longrightarrow}T'}
{(c)T\stackrel{\overline{a}(c)}{\longrightarrow}T'}\ \ \ \ \  &
\inference{T\stackrel{\lambda}{\longrightarrow}T'}
{(c)T\stackrel{\lambda}{\longrightarrow}(c)T'}\ c\not\in n(\lambda)
\end{array}\]

\noindent{\em Condition}
\[\begin{array}{cc}
\inference{T\stackrel{\lambda}{\longrightarrow}T'}{[a{=}a]T\stackrel{\lambda}{\longrightarrow}T'}\
\ \ \ \ &
\inference{T\stackrel{\lambda}{\longrightarrow}T'}{[a{\not=}b]T\stackrel{\lambda}{\longrightarrow}T'}
\end{array}\]

\noindent{\em Replication}
\[\begin{array}{cc}
\inference{}{!\overline{a}b.T\stackrel{\overline{a}b}{\longrightarrow}T\,|\,!\overline{a}b.T}\
\ \ \ \  &
\inference{}{!a(x).T\stackrel{ab}{\longrightarrow}T\{b/x\}\,|\,!a(x).T}
\end{array}\]
\vspace*{1mm}

The first composition rule takes care of the structural property. The
second and the third define interactions. Their symmetric versions
have been omitted. Particular attention should be paid to the action
rules. An input action may receive a name from another term. It is
not supposed to accept a name variable. This is because an output
action is allowed to release a name, not a name variable. To go
along with this semantics of interaction, the rule for the mismatch
operator is defined accordingly. Since distinct names are different
constant names, the condition $a{\not=}b$ is equivalent to $\top$.
The transition
$[x{\not=}c]\overline{a}b.T\stackrel{\overline{a}b}{\longrightarrow}T$
for instance is not admitted since the name variable $x$ needs to be
instantiated before the mismatch can be evaluated. One advantage of
this semantics is the validity of the following lemma.

\begin{lemma}\label{2008-09-25}
The following statements are valid whenever
$S\stackrel{\lambda}{\longrightarrow}T$.
\begin{enumerate}
\item
$S\alpha\stackrel{\lambda\alpha}{\longrightarrow}T\alpha$ for
every renaming $\alpha$.
\item
$S\sigma\stackrel{\lambda}{\longrightarrow}T\sigma$ for every
substitution $\sigma$.
\item
$S\rho\stackrel{\lambda}{\longrightarrow}T\rho$ for every assignment
$\rho$.
\end{enumerate}
\end{lemma}

The operational semantics of a process calculus draws a sharp line
between internal actions ($\tau$-actions) and external actions
(non-$\tau$-actions). From the point of view of interaction, the
former is unobservable and the latter is observable. A {\em complete
internal action sequence} of a process $P$ is either an infinite
$\tau$-action sequence
\[P\stackrel{\tau}{\longrightarrow}P_{1}\stackrel{\tau}{\longrightarrow}\ldots\stackrel{\tau}{\longrightarrow}P_{i}\stackrel{\tau}{\longrightarrow}\ldots\]
or a finite $\tau$-action sequence
$P\stackrel{\tau}{\longrightarrow}P_{1}\stackrel{\tau}{\longrightarrow}\ldots\stackrel{\tau}{\longrightarrow}P_{n}$,
where $P_{n}$ cannot perform any $\tau$-action. A process
$P$ is {\em divergent} if it has an infinite internal action
sequence; it is {\em terminating} otherwise.

In the $\pi$-calculus with the uniform treatment of names,
$\textit{if}\_\textit{then}\_\textit{else}\_$ command is defined
with the help of the unguarded choice operator. A nice thing about
the present semantics is that one may define
$\textit{if}\_\textit{then}\_\textit{else}\_$ command in the
following manner.
\begin{eqnarray*}
\textit{if}\;m{=}n\;\textit{then}\;S\;\textit{else}\;T
&\stackrel{\rm def}{=}& [m{=}n]S\,|\,[m{\not=}n]T.
\end{eqnarray*}
The idea can be generalized. Suppose $\{\varphi_{i}\}_{1\le i\le n}$
is a finite collection of {\em disjoint conditions}, meaning that
$\models\varphi_{i}\wedge\varphi_{j}\Rightarrow\bot$ for all
$i,j\le n$ such that $i\not=j$. The {\em conditional choice}
$\sum_{i\in\{1,..,n\}}\varphi_{i}T_{i}$ is defined in the following
fashion.
\begin{eqnarray}\label{choice-1}
\sum_{1\le i\le n}\varphi_{i}T_{i} &\stackrel{\rm def}{=}&
\varphi_{1}T_{1}\,|\,\ldots\,|\,\varphi_{n}T_{n}.
\end{eqnarray}
In practice most choice operations are actually conditional
choices~\cite{CaiFu2011}. Another form of choice is the
so-called {\em internal choice}, defined as follows:
\begin{eqnarray}\label{choice-2}
\sum_{1\le i\le n}\psi_{i}\tau.T_{i} &\stackrel{\rm def}{=}&
(c)(\psi_{1}c.T_{1}\,|\,\ldots\,|\,\psi_{n}c.T_{n}\,|\,\overline{c}),
\end{eqnarray}
where $\{\psi_{i}\}_{1\le i\le n}$ is a collection of conditions and
$c$ appears in none of $T_{1},\ldots,T_{n}$.

In~\cite{FuLu2010} it is shown that the replication, the fixpoint
operation and the parametric definition are equivalent in the $\pi$-calculus, as long as all the choices are guarded. The fixpoint
operator plays an indispensable role in proof systems for regular
processes. The parametric definition is strictly more powerful in
some variants of the $\pi$-calculus. Proof systems for regular
$\pi$-processes will not be a major concern of this paper. In all
the qualified name-passing calculi studied in the present paper the
parametric definition is equivalent to the replicator. We shall
therefore ignore both the fixpoint and the parametric definition in
the rest of the paper.

Some more notation and terminology need be defined. Given an
action $\lambda$, we may define $\overline{\lambda}$ as follows:
\begin{eqnarray*}
\overline{\lambda} &\stackrel{\rm def}{=}& \left\{
\begin{array}{ll}
\overline{a}b, & \mathrm{if}\ \lambda=ab, \\
ab, & \mathrm{if}\ \lambda=\overline{a}b, \\
ab, & \mathrm{if}\ \lambda=\overline{a}(b), \\
\tau, & \mathrm{if}\ \lambda=\tau.
\end{array}
\right.
\end{eqnarray*}
The notation $\overline{\ell}$ should be understood accordingly. We
shall write $\widetilde{\_}$ for a finite sequence of something of
same kind. For example a finite sequence of the names
$c_{1},\ldots,c_{n}$ can be abbreviated to $\widetilde{c}$. Let
$\Longrightarrow$ be the reflexive and transitive closure of
$\stackrel{\tau}{\longrightarrow}$. We write
$\stackrel{\widehat{\lambda}}{\longrightarrow}$ for the syntactic equality $\equiv$ if
$\lambda=\tau$ and for $\stackrel{\lambda}{\longrightarrow}$ otherwise. The notation
$\stackrel{\widehat{\lambda}}{\Longrightarrow}$ stands for
$\Longrightarrow\stackrel{\widehat{\lambda}}{\longrightarrow}\Longrightarrow$.
If $\lambda^{*}=\lambda_{1}\ldots\lambda_{n}$, we write
$P\stackrel{\lambda^{*}}{\Longrightarrow}$ if
$P\stackrel{\lambda_{1}}{\Longrightarrow}P_{1}\ldots
\stackrel{\lambda_{n}}{\Longrightarrow}P_{n}$ for some
$P_{1},\ldots,P_{n}$. If
$P\stackrel{\lambda^{*}}{\Longrightarrow}P'$, we say that $P'$ is a
{\em descendant} of $P$. Notice that $P$ is a descendant of itself.

In sequel a relation always means a binary relation on the processes
of the $\pi$-calculus or one of its variants. If $\mathcal{R}$ is a
relation, $\mathcal{R}^{-1}$ is the reverse of $\mathcal{R}$ and
$P\mathcal{R}Q$ for membership assertion. If $\mathcal{R}'$ is
another relation, the composition $\mathcal{R};\mathcal{R}'$ is the
relation $\{(P,Q) \mid \exists L.P\mathcal{R}L\wedge
L\mathcal{R}'Q\}$.

\subsection{Variants}

A number of `subcalculi' of $\pi$ have been studied. These variants
are obtained by omitting some operators. They can also be obtained
by restricting the use of the received names, or the use of
continuation, or the forms of prefix etc.. In this section we give a
brief summary of some of the variants. Our definitions of the
variants are slightly different from the popular ones, since we
attempt to give a systematic classification of the $\pi$-variants.

The guarded choice is a useful operator in encoding.
It is also a basic operator in proof systems. The independence of the choice operator from the
other operators of the $\pi$-calculus is established in for
example~\cite{Palamidessi2003,FuLu2010}. A lot of programming can be
carried out in the $\pi$-calculus using prefix terms rather than the guarded
choice terms~\cite{Walker1991,Walker1995}. We will write $\pi^{-}$
for the subcalculus of $\pi$ obtained by replacing the guarded
choice terms by the prefix terms of the form $\pi.T$. In many
aspects $\pi^{-}$ is just as good as $\pi$. It is for example
complete in the sense that one may embed the recursion
theory~\cite{Rogers1987} in $\pi^{-}$. See~\cite{FuYuxi} for detail.
The calculus $\pi^{-}$ can be further slimmed down by removing the
match and the mismatch operators. We shall call it the {\em minimal
$\pi$-calculus} and shall denote it by $\pi^{M}$. It is conjectured
that $\pi^{M}$ is considerably weaker than $\pi^{-}$. The intuition
behind the conjecture is that the if-command cannot be faithfully
encoded in $\pi^{M}$.

Some of the syntactical simplifications of the $\pi$-calculus have been
studied in literature. Here are some of them:
\begin{itemize}
\item Merro and Sangiorgi's local $\pi$-calculus forbids the use of a
received name as an input
channel~\cite{Merro2000,MerroSangiorgi2004}. The word `local' refers
to the fact that a receiving party may only use a received local
name to call upon subroutines delivered by the sending party. The
calculus can be seen as a theoretical foundation of the concurrent
and distributed languages like {\em Pict}~\cite{PierceTurner2000} and
{\em Join}~\cite{FournetGonthier1996}. The local variants we
introduce in this paper are obtained from the $\pi$-calculus by
restricting the use of the received names. A received name may be
used as an output channel or an input channel. It may also appear in
the object position. This suggests the following three variants. In
$\pi^{L}$ the grammar for prefix is
\begin{eqnarray*}
\pi &:=& a(x) \mid \overline{n}m.
\end{eqnarray*}
In $\pi^{R}$ the syntax of prefix is
\begin{eqnarray*}
\pi &:=& n(x) \mid \overline{a}m,
\end{eqnarray*}
and in $\pi^{S}$ it is
\begin{eqnarray*}
\pi &:=& n(x) \mid \overline{n}c.
\end{eqnarray*}
In $\pi^{S}$ only control information may be communicated, data
information may not be passed around. A piece of control information
may be sent once, it may not be resent.
\end{itemize}
Other well known variants are syntactical simplifications of
$\pi^{-}$. Let's see a couple of them.
\begin{itemize}
\item The $\pi$-process $!a(x).T$ can be seen as a method that can be
invoked by a process of the form $\overline{a}c.S$. This object
oriented programming style is typical of the name-passing calculi.
The object may invoke the method several times. It follows that the
method must be perpetually available. The minimal modification of
$\pi$ that embodies this idea is $\pi^{O}$-calculus with following
grammar:
\begin{eqnarray*}
S,T &:=& {\bf 0} \mid \overline{n}m.T \mid S\,|\,T \mid (c)T \mid [p{=}q]T \mid [p{\not=}q]T \mid \;!n(x).T.
\end{eqnarray*}
The superscript $O$ refers to the fact that $\pi^{O}$-calculus has
only proper output prefix; it also reminds of the object oriented
programming style. The dual of $\pi^{O}$-calculus,
$\pi^{I}$-calculus, is defined by the following grammar.
\begin{eqnarray*}
S,T &:=& {\bf 0} \mid n(x).T \mid S\,|\,T \mid (c)T \mid [p{=}q]T \mid [p{\not=}q]T \mid \;!\overline{n}m.T.
\end{eqnarray*}
The relationship between $\pi^{I}$ and $\pi^{O}$ is intriguing.
There is a straightforward encoding $\llbracket\_\rrbracket^{\iota}$
from the former to the latter defined as follows:
\begin{eqnarray*}
\llbracket n(x).T\rrbracket^{\iota} &\stackrel{\rm def}{=}&
\overline{n}(c).!c(x).\llbracket T\rrbracket^{\iota}, \\
\llbracket !\overline{n}m.T\rrbracket^{\iota} &\stackrel{\rm
def}{=}& !n(x).\overline{x}m.\llbracket T\rrbracket^{\iota}.
\end{eqnarray*}
The soundness of this encoding is unknown. Both $\pi^{O}$ and
$\pi^{I}$ have too weak control power to be considered as proper
models. In $\pi^{O}$ for example, a name cannot be used to input
two ordered pieces of information since a sending party would get
confused. There is no way to define for example a polyadic prefix
term like $a(x,y).T$. We shall not consider these two variants in
the rest of the paper.

\item Honda and Tokoro's object
calculus~\cite{HondaTokoro1991ObjCal,HondaTokoro1991AsynSem,HondaYoshida1995}
and Boudol's asynchronous $\pi$-calculus~\cite{Boudol1992} are based
on the idea that an output prefix with a continuation does not
interact with any environment. Instead it evolves into a process
representing the continuation and an atomic output process whose
sole function is to send a name to an environment. This is captured
by the following transitions.
\begin{eqnarray}
\overline{a}c.T &\stackrel{\tau}{\longrightarrow}&
\overline{a}c\,|\,T, \label{2009-06-15AO} \\
\overline{a}c\,|\,a(x).T &\stackrel{\tau}{\longrightarrow}& {\bf
0}\,|\,T\{c/x\}. \label{2009-06-15AI}
\end{eqnarray}
(\ref{2009-06-15AI}) is absolutely necessary, whereas
(\ref{2009-06-15AO}) can be avoided in view of the fact that
$\overline{a}c.T$ must be equal to $\overline{a}c\,|\,T$. Moreover
if we think asynchronously the replication, match and mismatch
operators should not apply to the processes of the form
$\overline{n}m$. This explains the following grammar of $\pi^{A}$:
\begin{eqnarray*}
S,T &:=& {\bf 0} \mid \overline{n}m \mid \sum_{i\in I}\psi_{i}n_{i}(x).T \mid S\,|\,T \mid (c)T \mid \;!n(x).T,
\end{eqnarray*}
where $\psi$ is a finite sequence of match/mismatch operations. The
standard prefix operators can be mimicked in $\pi^{A}$ in the
following way~\cite{Boudol1992}:
\begin{eqnarray}
\llbracket p(x).S \rrbracket &\stackrel{\rm def}{=}& p(u).(d)(\overline{u}d\,|\,d(x).\llbracket S \rrbracket), \label{asyn2synin}\\
\llbracket \overline{p}q.T \rrbracket &\stackrel{\rm def}{=}&
(c)(\overline{p}c\,|\,c(v).(\overline{v}q\,|\,\llbracket T
\rrbracket)). \label{asyn2synout}
\end{eqnarray}
This encoding of the output prefix is however not very robust from
the observational point of view. In fact~\cite{CacciagranoCorradiniPalamidessi2006} have proved
that no encodings of the output prefix exist that preserve must
equivalence. If divergence is not seen as an important issue, then
there are interesting encodings into the asynchronous $\pi$-calculus
as shown in~\cite{NestmannPierce1996} and in~\cite{Nestmann2000}.
The asynchronous calculi have a very different flavor from the
synchronous calculi. We will not discuss $\pi^{A}$ in the rest of
the paper. But see~\cite{FuYuxi-Theory-by-Process} for an alternative
approach to the asynchronous process calculi.

\item A more distant relative is Sangiorgi's {\em private
$\pi$-calculus}~\cite{SangiorgiWalker2001BOOK}, initially called
$\pi I$-calculus~\cite{Sangiorgi1996TCS}. The grammar of $\pi^{P}$
differs from that of $\pi$ in the definition of prefix. In $\pi^{P}$
it is given by the following BNF.
\begin{eqnarray*}
\pi &:=& n(x) \mid \overline{n}(c).
\end{eqnarray*}
Unlike all the above variants, the $\pi^{P}$-calculus has a
different action set than the $\pi$-calculus. There are no free output
actions. One has typically that
\[a(x).S\,|\,\overline{a}(c).T \stackrel{\tau}{\longrightarrow} (c)(S\{c/x\}\,|\,T).\]
In $\pi^{P}$ the name emitted by an output action is always local.
It is worth remarking that this particular $\pi^{P}$-calculus is not
equivalent to the version in which replications are abolished in
favor of parametric definitions~\cite{Sangiorgi1996TCS,FuLu2010}. It
is interesting to compare $\pi^{P}$-calculus with
$\pi^{S}$-calculus. In both models the messages communicated at run
time are foreordained at compile time. The difference is that in
$\pi^{P}$-calculus only local messages can be released, whereas in
$\pi^{S}$-calculus global messages may also be transmitted. Despite
the results established in~\cite{Boreale1996}, the variant $\pi^{P}$
appears much less expressive than $\pi^{S}$. For one thing the match
and mismatch operators in $\pi^{P}$ are redundant. So it is more
precise to say that $\pi^{P}$ is a variant of $\pi^{M}$. Since
$\pi^{P}$ has a different action set than $\pi^{M}$, we shall not discuss it in this paper.
\end{itemize}

One may combine the restrictions introduced in the above variants to
produce even more restricted `subcalculi'. Some of these
`subcalculi' can be rejected right away. For example the variant in
which the received names can only be used as data is equivalent to CCS. Another counter
example is the variant in which both the input capability and the
output capability are perpetual. This calculus is not very useful
since no computations in this model ever terminate. It simply does
not make sense to talk about Turing completeness for this particular
calculus. Merro and Sangiorgi's local $\pi$-calculus is the
asynchronous version of $\pi^{L}$ without the match/mismatch
operators. In~\cite{PalamidessiSaraswatValenciaVictor2006,CacciagranoCorradiniArandaValencia2008} the authors look at some
asynchronous versions of $\pi^{I}$ and $\pi^{O}$. Their work
puts more emphasis on the perpetual availability of the input/output
actions.

A crucial question can be asked about each of the variants: which are complete?
At the moment we know that $\pi,\pi^{L},\pi^{R},\pi^{S},\pi^{M}$ are complete.
See~\cite{FuYuxi} for a formal definition of completeness and~\cite{XueLongFu2011} for additional discussions and proofs.

\section{Equality}\label{sec-Equality}

The starting point of observational theory is to do with the
definition of process equality. It was realized from the very
beginning~\cite{Milner1980,Sangiorgi2009} that the equalities for
processes ought to be observational since it is the effects that
processes have in environments that really matter. Different
application platforms may demand different observing powers, giving
rise to different process equalities. This interactive
viewpoint~\cite{Milner1993} lies at the heart of the theory of
equality. In a distributed environment for example, a process is
subject to perpetual interventions from a potentially unbounded
number of observers in an interleaving manner. Here the right
equivalences are bisimulation equivalences. In a one shot testing
scenario, equivalent processes are indistinguishable by any single
tester in an exclusive fashion. It is nobody's business what the
equivalent processes would turn into after a test. This is the
testing equivalence. In this section we take a look at both the
bisimulation approach and the testing approach.

\subsection{Absolute Equality}

A reasonable criterion for equality of {\em self evolving}
objects is the famous bisimulation property of Milner~\cite{Milner1989}
and Park~\cite{Park1981}.

\begin{definition}\label{weak-bisimulation}
A relation $\mathcal{R}$ is a {\em weak bisimulation} if the
following statements are valid.
\begin{enumerate}
\item  If $Q\mathcal{R}^{-1}P\stackrel{\tau}{\longrightarrow}P'$ then
$Q\Longrightarrow Q'\mathcal{R}^{-1}P'$ for some $Q'$.
\item  If $P\mathcal{R}Q\stackrel{\tau}{\longrightarrow}Q'$ then
$P\Longrightarrow P'\mathcal{R}Q'$ for some $P'$.
\end{enumerate}
\end{definition}
The bisimulation property captures the intuition that an equivalence
for self evolving objects should be {\em maintainable} by the
equivalent objects themselves when left alone. It is no good if two
self evolving objects were equivalent one minute ago and will have
to evolve into inequivalent objects one minute later. Bisimulation
is about self evolving activities, not about interactive activities.
What more can be said about self evolution? An evolution path may
experience a series of state changes. At any particular time of the
evolution, the process may turn into an equivalent state, and it may
also move to an inequivalent state. Inequivalent states generally
exert different effects on environments, while equivalent states
always have the same interactive behaviors.
Definition~\ref{weak-bisimulation} can be criticized in that it
overlooks the important difference between these two kinds of state
transitions. The rectification of Definition~\ref{weak-bisimulation}
is achieved in~\cite{vanGlabbeekWeijland1989-first-paper-bb} by the
branching bisimulation. The following particular formulation is
given by Baeten~\cite{Baeten1996}.

\begin{definition}\label{bisimulation}
A binary relation $\mathcal{R}$ is a {\em bisimulation} if it
validates the following {\em bisimulation property}.
\begin{enumerate}
\item  If $Q\mathcal{R}^{-1}P\stackrel{\tau}{\longrightarrow}P'$ then one of the
following statements is valid.
\begin{enumerate}
\item $Q\Longrightarrow Q'$ for some $Q'$ such that $Q'\mathcal{R}^{-1}P$ and $Q'\mathcal{R}^{-1}P'$.

\item $Q\Longrightarrow
Q''\mathcal{R}^{-1}P$ for some $Q''$ such that
$Q''\stackrel{\tau}{\longrightarrow}Q'\mathcal{R}^{-1}P'$ for some
$Q'$.
\end{enumerate}

\item  If $P\mathcal{R}Q\stackrel{\tau}{\longrightarrow}Q'$ then one of the
following statements is valid.
\begin{enumerate}
\item $P\Longrightarrow P'$ for some $P'$ such that $P'\mathcal{R}Q$ and $P'\mathcal{R}Q'$.

\item $P\Longrightarrow
P''\mathcal{R}Q$ for some $P''$ such that
$P''\stackrel{\tau}{\longrightarrow}P'\mathcal{R}Q'$ for some $P'$.
\end{enumerate}
\end{enumerate}
\end{definition}
Since we take the branching bisimulations as {\em the}
bisimulations, we leave out the adjective.

A process not only evolves on its own, it also interacts with other
processes. A plain criterion for the equalities of {\em interacting}
objects is that an equality between two objects should be {\em
maintainable} no matter how environments may interact with them. It
is extremely important to get it right what an environment can and
cannot do. In a distributed scenario, an environment may interact
with a process placed in the regional network. An environment is not
capable of grabbing a {\em running} process and reprogram it as it
were. So it would not be appropriate to admit, say
$a(x).(\_\{x/a\}\,|\,O)$, as an environment.

\begin{definition}
An {\em environment} is a process of the form
$(\widetilde{c})(\_\,|\,O)$ with the specified hole indicated by
`$\_$'.
\end{definition}

A prerequisite for two processes $P,Q$ to be observationally
equivalent is that no environment can tell them apart. But saying
that the environment $(\widetilde{c})(\_\,|\,O)$ cannot observe any
difference between $P$ and $Q$ is nothing but saying that
$(\widetilde{c})(P\,|\,O)$ and $(\widetilde{c})(Q\,|\,O)$ are
observationally equivalent. Hence the next definition.

\begin{definition}
A relation $\mathcal{R}$ is {\em extensional} if the following
statements are valid.
\begin{enumerate}
\item  If $L\mathcal{R}M$ and $P\mathcal{R}Q$ then
$(L\,|\,P)\;\mathcal{R}\;(M\,|\,Q)$.
\item  If $P\mathcal{R}Q$ then $(c)P\;\mathcal{R}\;(c)Q$.
\end{enumerate}
\end{definition}

Process equivalences are {\em observational}. A process $P$ is {\em
observable}, written $P{\Downarrow}$, if
$P\Longrightarrow\stackrel{\ell}{\longrightarrow}P'$ for some $P'$.
It is {\em unobservable}, written $P{\not\Downarrow}$, if it is not
observable. Occasionally we write $P{\Downarrow_{\ell}}$ for
$\exists P'.P\Longrightarrow\stackrel{\ell}{\longrightarrow}P'$ and
similarly $P{\downarrow_{\ell}}$ for $\exists
P'.P\stackrel{\ell}{\longrightarrow}P'$. It should be apparent that
two equivalent processes must be both observable or both
unobservable.

\begin{definition}
A relation $\mathcal{R}$ is {\em equipollent} if
$P{\Downarrow}\Leftrightarrow Q{\Downarrow}$ whenever
$P\mathcal{R}Q$.
\end{definition}

Equipollence is the weakest condition that an observational
equivalence has to satisfy. Now suppose $P$ and $Q$ are
observationally inequivalent. If we ignore the issue of divergence,
then there must exist some $\ell$ such that $P{\Downarrow_{\ell}}$
and for all $\ell'$ such that $Q{\Downarrow_{\ell'}}$, the actions
$\ell$ and $\ell'$ exert different effects on some environment
$(\widetilde{c})(\_\,|\,O)$. As we mentioned just now, what this
means is that $(\widetilde{c})(P\,|\,O)$ and
$(\widetilde{c})(Q\,|\,O)$ are observationally inequivalent. So we
may repeat the argument for $(\widetilde{c})(P\,|\,O)$ and
$(\widetilde{c})(Q\,|\,O)$. But if the calculus is strong enough,
say it is complete, then the inequivalence eventually boils down to
the fact that one delivers a result at some name and the other fails
to do so at any name. Notice that the localization operator plays a
crucial role in this argument.

The theory of process calculus has been criticized for not paying
enough attention to the issue of divergence. This criticism is more
or less to the point. If the theory of process calculus is meant to
be an extension of the theory of computation, a divergent
computation must be treated in a different way than a terminating
computation. How should we formulate the requirement that process
equivalences should be divergence respecting. A property that is not
only divergence preserving but also consistent with the idea of
bisimulation is codivergence.

\begin{definition}
A {\em relation} is \emph{codivergent} if $P\mathcal{R}Q$ implies
the following property.
\begin{enumerate}
  \item  If $Q \stackrel{\tau}{\longrightarrow} Q_1 \stackrel{\tau}{\longrightarrow} \cdots \stackrel{\tau}{\longrightarrow} Q_n \stackrel{\tau}{\longrightarrow} \cdots$
   is an infinite internal action sequence, then there must be some $k \geq 1$ and
    $P'$ such that $P \stackrel{\tau}{\Longrightarrow} P' \;\mathcal{R}\; Q_k$.
  \item  If $P \stackrel{\tau}{\longrightarrow} P_1 \stackrel{\tau}{\longrightarrow} \cdots \stackrel{\tau}{\longrightarrow} P_n \stackrel{\tau}{\longrightarrow} \cdots$
   is an infinite internal action sequence, then there must be some $k \geq 1$ and
    $Q'$ such that $Q \stackrel{\tau}{\Longrightarrow} Q' \;\mathcal{R}^{-1}\; P_k$.
\end{enumerate}
\end{definition}

Using the codivergence condition, one may define an equivalence that
is advocated in~\cite{FuYuxi} as {\em the} equality for process
calculi.

\begin{definition}\label{absolute-equality}
The {\em absolute equality} $=$ is the largest relation such that
the following statements are valid.
\begin{enumerate}
\item  It is reflexive.
\item  It is equipollent, extensional, codivergent and bisimilar.
\end{enumerate}
\end{definition}
Alternatively we could define the absolute equality as the largest
equipollent codivergent bisimulation that is closed under the
environments.

The virtue of Definition~\ref{absolute-equality} is that it is
completely model independent, as long as we take the view that the
composition operator and the localization operator are present in
all process calculi. From the point of interaction, there cannot be
any argument against equipollence and extensionality. From the point
of view of computation, there is no question about codivergence and
bisimulation. The only doubt one may raise is if
Definition~\ref{absolute-equality} is strong enough. We shall
demonstrate in this paper that as far as the $\pi$-calculus is
concerned, the absolute equality is the most appropriate equivalence
relation.

This is the right place to state an extremely useful technical
lemma~\cite{FuYuxi1999}, the Bisimulation Lemma. Although worded for
$=$, Bisimulation Lemma, called X-property by De Nicola, Montanari and Vaandrager~\cite{DeNicolaMontanariVaandrager1990}, is actually
valid for all the {\em observational} equivalences.

\begin{lemma}\label{bisimulation-lemma}
If $P\Longrightarrow =Q$ and $Q\Longrightarrow =P$, then $P=Q$.
\end{lemma}

A distinguished property about the absolute equality is stated in
the next lemma. It is discovered
by van Glabbeek and Weijland~\cite{vanGlabbeekWeijland1989-first-paper-bb} for the branching
bisimilarity.

\begin{lemma}\label{inert}
If
$P_{0}\stackrel{\tau}{\longrightarrow}P_{1}\stackrel{\tau}{\longrightarrow}
\ldots\stackrel{\tau}{\longrightarrow}P_{n}=P_{0}$ then
$P_{0}=P_{1}=\ldots=P_{n}$.
\end{lemma}

A consequence of Lemma~\ref{inert} is that if
$P=Q\stackrel{\lambda}{\longrightarrow}Q'$, then any simulation
$P\stackrel{\tau}{\longrightarrow}P_{1}\stackrel{\tau}{\longrightarrow}
\ldots\stackrel{\tau}{\longrightarrow}P_{n}\stackrel{\lambda}{\longrightarrow}P'$
of $Q\stackrel{\lambda}{\longrightarrow}Q'$ by $P$ must satisfy the
property that $P=P_{1}=\ldots=P_{n}$. This phenomenon motivates the
following terminologies.
\begin{enumerate}
\item $P\stackrel{\tau}{\longrightarrow}P'$ is {\em deterministic}, notation $P\rightarrow P'$, if $P=P'$.

\item $P\stackrel{\tau}{\longrightarrow}P'$ is {\em nondeterministic}, notation $P\stackrel{\iota}{\longrightarrow}P'$, if $P\ne P'$.
\end{enumerate}
A deterministic internal action can be ignored.
A nondeterministic internal action has to be properly simulated.
The notation $\rightarrow^{*}$ stands for the reflexive and transitive closure of $\rightarrow$ and $\rightarrow^{+}$ for the transitive closure.
We will write $P\nrightarrow$ to indicate that $P\rightarrow P'$ for no $P'$.

Since the absolute equality applies not just to $\pi$ and its variants but to all process calculi, the
notation `$=$' could be confusing when more than one model is dealt with.
To remove the ambiguity, we sometimes write $=^{\mathbb{M}}$
for the absolute equality of model $\mathbb{M}$.
The same notation will be applied to other process equivalences.
In this paper, we write $\mathbb{V}$ for an arbitrary $\pi$-variant.

\subsubsection{External Bisimilarity}\label{subsec-External-Bisimulation}

The external bisimilarity requires that all observable actions are
explicitly simulated. This is commonly referred to as branching
bisimilarity if the codivergence is
dropped~\cite{vanGlabbeekWeijland1989-first-paper-bb}.

\begin{definition}\label{external-bi}
A codivergent bisimulation of $\mathbb{V}$ is a $\mathbb{V}$-{\em
bisimulation} if the following statements are valid for every
$\ell\in\mathcal{L}$.
\begin{enumerate}
\item  If $Q\mathcal{R}^{-1}P \stackrel{\ell}{\longrightarrow}P'$ then
$Q\Longrightarrow Q''\stackrel{\ell}{\longrightarrow
}Q'\mathcal{R}^{-1}P'$ and $P\mathcal{R}Q''$ for some $Q',Q''$.
\item  If $P\mathcal{R}Q \stackrel{\ell}{\longrightarrow}Q'$ then
$P\Longrightarrow P''\stackrel{\ell}{\longrightarrow
}P'\mathcal{R}Q'$ and $P''\mathcal{R}Q$ for some $P',P''$.
\end{enumerate}
The $\mathbb{V}$-{\em bisimilarity} $\simeq_{\mathbb{V}}$ is the
largest $\mathbb{V}$-bisimulation.
\end{definition}

$\mathbb{V}$-bisimilarity is a more effective counterpart of the
absolute equality $=_{\mathbb{V}}$. The latter provides the
intuition, while the former offers a tool for establishing
equational properties. The proof of following fact is a standard
textbook application of the bisimulation argument.

\begin{lemma}\label{algebraic}
The equivalence $\simeq_{\mathbb{V}}$ is closed under input choice,
output choice, composition, localization, match, mismatch, and
guarded replication.
\end{lemma}

An immediate consequence of Lemma~\ref{algebraic} is the inclusion described in the following lemma.

\begin{lemma}\label{linb}
$\simeq_{\mathbb{V}} \;\subseteq\;=_{\mathbb{V}}$.
\end{lemma}

If the processes of a model has strong enough observing power, then
the absolute equality is as strong as the external bisimilarity.
This is the case for the $\pi$-calculus. The proof of the following
theorem resembles a proof in~\cite{FuYuxi2005}.

\begin{theorem}\label{coincidence-theorem}
The $\pi$-bisimilarity $\simeq_{\pi}$ coincides with the absolute equality $=_{\pi}$.
\end{theorem}
\begin{proof}
In view of Lemma~\ref{linb}, we only have to prove
$=_{\pi}\;\subseteq\;\simeq_{\pi}$. Let $\mathcal{R}$ be the relation
\[\left\{(P,Q) \left| \begin{array}{l}
(c_{1},\ldots,c_{n})(\overline{a_{1}}c_{1}\,|\,\ldots\,|\,\overline{a_{n}}c_{n}\,|\,P)=_{\pi}
\\
(c_{1},\ldots,c_{n})(\overline{a_{1}}c_{1}\,|\,\ldots\,|\,\overline{a_{n}}c_{n}\,|\,Q),
 \\
\{a_{1},\ldots,a_{n}\}\cap gn(P\,|\,Q)=\emptyset, \ n\ge 0
\end{array} \right\}.\right.\]
To appreciate the relation notice that
$(c'c'')(\overline{a'}c'\,|\,\overline{a''}c''\,|\,P)\ne_{\pi}(c)(\overline{a'}c\,|\,\overline{a''}c\,|\,Q)$
for all $P,Q$ such that $\{a',a''\}\cap gn(P\,|\,Q)=\emptyset$. We
prove that $\mathcal{R}$ is an external bisimulation up to $\sim$,
where $\sim$ is defined in
Section~\ref{sec-Strong-Equality-and-Weak-Equality}. Now suppose
$A=_{\pi}B$ where
\begin{eqnarray*}
A &\stackrel{\rm def}{=}&
(c_{1},\ldots,c_{n})(\overline{a_{1}}c_{1}\,|\,\ldots\,|\,\overline{a_{n}}c_{n}\,|\,P),
\\
B &\stackrel{\rm def}{=}&
(c_{1},\ldots,c_{n})(\overline{a_{1}}c_{1}\,|\,\ldots\,|\,\overline{a_{n}}c_{n}\,|\,Q),
\end{eqnarray*}
such that $\{a_{1},\ldots,a_{n}\}\cap gn(P\,|\,Q)=\emptyset$ and
$n\ge 0$. Consider an action $P\stackrel{\ell}{\longrightarrow}P'$
of $P$. There are four cases.
\begin{enumerate}
\item $\ell=ab$. Let $C$ be $\overline{a}b+\overline{a}c.d$ for some fresh name $c$. By equipollence and extensionality,
$A\,|\,C\stackrel{\tau}{\longrightarrow}A'\,|\,\mathbf{0}$ must be
matched up by $B\,|\,C\Longrightarrow
B''\,|\,C\stackrel{\tau}{\longrightarrow}B'\,|\,\mathbf{0}$ for some
$B',B''$. Since
$B''\,|\,C\stackrel{\tau}{\longrightarrow}B'\,|\,\mathbf{0}$ must be
a change of state, it must be the case that $A\,|\,C=_{\pi}B''\,|\,C$. It
follows easily from Bisimulation Lemma that
$B\rightarrow^{*}B''\stackrel{ab}{\longrightarrow}B'=_{\pi}A'$. Clearly
\begin{eqnarray*}
A' &\equiv&
(c_{1},\ldots,c_{n})(\overline{a_{1}}c_{1}\,|\,\ldots\,|\,\overline{a_{n}}c_{n}\,|\,P'),
\\
B' &\equiv&
(c_{1},\ldots,c_{n})(\overline{a_{1}}c_{1}\,|\,\ldots\,|\,\overline{a_{n}}c_{n}\,|\,Q'),
\\
B'' &\equiv&
(c_{1},\ldots,c_{n})(\overline{a_{1}}c_{1}\,|\,\ldots\,|\,\overline{a_{n}}c_{n}\,|\,Q''),
\end{eqnarray*}
for some $P',Q',Q''$. It follows from
$B\rightarrow^{*}B''\stackrel{ab}{\longrightarrow}B'$ that
$Q\Longrightarrow Q''\stackrel{ab}{\longrightarrow}Q'$. Moreover
$P\mathcal{R}Q''$ and $P'\mathcal{R}Q'$ by definition.

\item $\ell=\overline{a}b$ for some $b\not\in\{c_{1},\ldots,c_{n}\}$.
Let $c,d$ be fresh and let $D$ be defined by
$D \stackrel{\rm def}{=} a(x).[x{=}b]c+a(x).d$.
Now $A\,|\,D \stackrel{\tau}{\longrightarrow}
A'\,|\,[b{=}b]c$ must be simulated by
\[B\,|\,D \Longrightarrow  B''\,|\,D\stackrel{\tau}{\longrightarrow}B'\,|\,[b{=}b]c.\]
The rest of the argument is the same as in the previous case.

\item $\ell=\overline{a}(b)$ and $b\not\in\{c_{1},\ldots,c_{n}\}$.
Let $E$ be the following process
\[a(x).[x\notin gn(P\,|\,Q)]\overline{a_{n+1}}x+a(x).d,\]
where $d,a_{n+1}$ are fresh. The rest of the argument is similar.

\item $\ell=\overline{a}c_{i}$ for some $i\in\{1,\ldots,n\}$. Let
$F$ be the process
\[a(x).a_{i}(y).[x{=}y]\overline{a_{i}'}y+a(x).d,\]
where $d,a_{i}'$ are fresh. The rest of the argument is again similar.
\end{enumerate}
Additionally we need to consider the external actions of $P,Q$ that
communicate through one of the names $c_{1},\ldots,c_{n}$. There are
seven cases. In each case we are content with defining an
environment that forces external bisimulation.
\begin{enumerate}
\item $P\stackrel{c_{i}b}{\longrightarrow}P'$. Let $H$ be the
following process
\[a_{i}(z).(\overline{z}b.\overline{a_{i}'}z+\overline{z}d)\]
for some fresh $d,a_{i}'$.

\item $P\stackrel{c_{i}c_{j}}{\longrightarrow}P'$ such that $i\not=j$. Let $I$ be the
following process
\[a_{i}(z).a_{j}(y).(\overline{z}y.(\overline{a_{j}'}y\,|\,\overline{a_{i}'}z)+\overline{z}d)\]
for some fresh $d,a_{i}',a_{j}'$.

\item $P\stackrel{c_{i}c_{i}}{\longrightarrow}P'$. Let $J$ be the
following process
\[a_{i}(z).(\overline{z}z.\overline{a_{i}'}z+\overline{z}d)\]
for some fresh $d,a_{i}'$.

\item $P\stackrel{\overline{c_{i}}b}{\longrightarrow}P'$. Let $K$ be the
following process
\[a_{i}(z).(z(x).[x{=}b]\overline{a_{i}'}z+z(x).d)\]
for some fresh $d,a_{i}'$.

\item $P\stackrel{\overline{c_{i}}(c)}{\longrightarrow}P'$. Let $L$ be the
following process
\[a_{i}(z).(z(x).[x\notin gn(P\,|\,Q)](\overline{a_{n+1}}x\,|\,\overline{a_{i}'}z)+z(x).d),\]
where $d,a_{n+1},a_{i}'$ are fresh.

\item $P\stackrel{\overline{c_{i}}c_{j}}{\longrightarrow}P'$ such that $i\not=j$. Let $M$ be the
following process
\[a_{i}(z).(z(x).a_{j}(y).[x{=}y](\overline{a_{j}'}y\,|\,\overline{a_{i}'}z)+z(x).d)\]
for some fresh $d,a_{i}',a_{j}'$.

\item $P\stackrel{\overline{c_{i}}c_{i}}{\longrightarrow}P'$. Let $N$ be the
following process
\[a_{i}(z).(z(x).[x{=}z]\overline{a_{i}'}z+z(x).d)\]
for some fresh $d,a_{i}'$.
\end{enumerate}
We are done.
\end{proof}

We remark that the output choice is crucial in the above proof, but the input choice is not really necessary.
A slightly more complicated proof can be given without using any input choice.
At the moment we do not see how the theorem can be established without using the output choice.

The external characterization of the absolute equality for a particular model is an important issue.
It has been resolved for $\pi^{L},\pi^{R}$ in~\cite{XueLongFu2011}.
It is definitely a technically interesting but probably tricky exercise to remove the remaining question marks in Fig.~\ref{external-characterization}.

\begin{figure}[t]
\begin{center}
\begin{tabular}{|c|c|c|c|c|c|c|}
\hline
 $\pi$ & $\pi^{-}$ & $\pi^{M}$ & $\pi^{L}$ & $\pi^{R}$ & $\pi^{S}$ \\
\hline\hline
 $\surd$ & ? & ? & $\surd$ & $\surd$ & ? \\
\hline
\end{tabular}
\end{center}
\caption{External Characterization of $=_{\mathbb{V}}$.
\label{external-characterization}}
\end{figure}

\begin{problem}\label{ext=abs}
What are the answers to the questions raised in
Fig.~\ref{external-characterization}?
\end{problem}

There is a standard way to extend the absolute equality from
processes to terms. The following definition is well-known.

\begin{definition}\label{2009-01-25}
Let $\asymp$ be a process equivalence.
Then $S\asymp T$ if $S\rho\asymp T\rho$ for every assignment $\rho$ such that $fv(S\,|\,T)\subseteq dom(\rho)$.
\end{definition}

The equivalence $=$ on $\pi$-terms satisfies two crucial equalities.
One is
\begin{equation}\label{open-eq-1}
T = [x{=}y]T\,|\,[x{\not=}y]T.
\end{equation}
The other is
\begin{equation}\label{open-eq-2}
[x{\not=}y]\lambda.T = [x{\not=}y]\lambda.[x{\not=}y]T.
\end{equation}
Both equalities explain the role of the name variables. It is the
viewpoint of this paper that all process equivalences should
validate both (\ref{open-eq-1}) and (\ref{open-eq-2}).

\subsubsection{Strong Equality and Weak Equality}\label{sec-Strong-Equality-and-Weak-Equality}

The only way to refine the absolute equality in a model independent
way is to strengthen the bisimulation property. Among all the
refinements of Definition~\ref{bisimulation}, the most well known
one is Milner and Park's strong
bisimulation~\cite{Park1981,Milner1989}.

\begin{definition}
A symmetric relation $\mathcal{R}$ is a {\em strong bisimulation} if
$P\stackrel{\tau}{\longrightarrow}P'\mathcal{R}Q'$ for some $P'$
whenever $P\mathcal{R}Q\stackrel{\tau}{\longrightarrow}Q'$.
\end{definition}
It is clear that a strong bisimulation is automatically codivergent.
Hence the next definition.
\begin{definition}
The {\em strong equality} $=^{s}$ is the largest reflexive,
equipollent, extensional, strong bisimulation.
\end{definition}
The strong equality is the same as the strong bisimilarity
of Milner~\cite{Milner1989}.

\begin{definition}\label{external-strong-bi}
A strong bisimulation is an {\em external strong bisimulation} if
the following statements are valid for every $\ell\in\mathcal{L}$.
\begin{enumerate}
\item  If $Q\mathcal{R}^{-1}P \stackrel{\ell}{\longrightarrow}P'$ then
$Q\stackrel{\ell}{\longrightarrow}Q'\mathcal{R}^{-1}P'$ for some
$Q'$.
\item  If $P\mathcal{R}Q
\stackrel{\ell}{\longrightarrow}Q'$ then
$P\stackrel{\ell}{\longrightarrow}P'\mathcal{R}Q'$ for some $P'$.
\end{enumerate}
The {\em external strong bisimilarity} $\sim$ is the largest
external strong bisimulation.
\end{definition}

One finds a very useful application of the strong equality in the
`bisimulation up to $\sim$' technique~\cite{SangiorgiMilner1992}.
The strong equality equates all the structurally equivalent
processes. Some well known equalities are stated below.

\begin{lemma}\label{monoidal}
The following equalities are valid.
\begin{enumerate}
\item  $P\,|\,{\bf 0}\sim P$; $P\,|\,Q\sim Q\,|\,P$;
$P\,|\,(Q\,|\,R)\sim (P\,|\,Q)\,|\,R$.

\item  $(c){\bf 0}\sim {\bf 0}$; $(c)(d)P\sim (d)(c)P$.

\item  $(c)(P\,|\,Q)\sim P\,|\,(c)Q$ if $c\not\in gn(P)$.

\item  $[a{=}a]P\sim [a{\not=}b]P\sim P$; $[a{\not=}a]P\sim
[a{=}b]P\sim {\bf 0}$.

\item  $!\pi.P\sim \pi.P\,|\,!\pi.P$.
\end{enumerate}
\end{lemma}

In the other direction we may weaken the absolute equality by
relaxing the bisimulation property \textit{a la}
Definition~\ref{weak-bisimulation}.

\begin{definition}
The {\em weak equality} $=^{w}$ is the largest reflexive,
equipollent, extensional, codivergent, weak bisimulation.
\end{definition}

The external characterization of the weak equality is Milner's weak
bisimilarity~\cite{Milner1989} enhanced with the codivergence
condition.

\begin{definition}\label{external-weak-bi}
A codivergent bisimulation is an {\em external weak bisimulation} if
the following statements are valid for every $\ell\in\mathcal{L}$.
\begin{enumerate}
\item  If $Q\mathcal{R}^{-1}P \stackrel{\ell}{\longrightarrow}P'$ then
$Q\stackrel{\ell}{\Longrightarrow}Q'\mathcal{R}^{-1}P'$ for some
$Q'$.

\item  If $P\mathcal{R}Q \stackrel{\ell}{\longrightarrow}Q'$ then
$P\stackrel{\ell}{\Longrightarrow}P'\mathcal{R}Q'$ for some $P'$.
\end{enumerate}
The {\em external weak bisimilarity} $\simeq $ is the largest
external weak bisimulation.
\end{definition}

The proof of Theorem~\ref{coincidence-theorem} can be extended to a
proof of the following coincidence result.

\begin{proposition}\label{2009-09-23}
The equivalences $=^{w},\simeq $ coincide in $\pi^{M}$, $\pi^{-}$
and $\pi$.
\end{proposition}

The relationships between the strong equality, the absolute equality
and the weak equality are stated in the next proposition.

\begin{proposition}
The inclusions $=^{s} \;\subsetneq\;=\;\subsetneq\;=^{w}$ are
strict.
\end{proposition}
\begin{proof}
It is well known that $a.(\tau.b+\tau.c)=^{w}a.(\tau.b+\tau.c)+a.c$ holds but
$a.(\tau.b+\tau.c)=a.(\tau.b+\tau.c)+a.c$ is not valid.
\end{proof}

\subsubsection{Remark}

In the theory of bisimulation semantics, the three major
contributions are the bisimulation of Milner~\cite{Milner1989}
and Park~\cite{Park1981}, the branching bisimulation
of van Glabbeek and Weijland~\cite{vanGlabbeekWeijland1989-first-paper-bb}, and the barbed
bisimulation of Milner and Sangiorgi~\cite{MilnerSangiorgi1992}. Despite its nice
properties~\cite{DeNicolaMontanariVaandrager1990,vanGlabbeek1994,DeNicolaVaandrager1995,BaierHermanns1997},
the branching bisimulation is not as widely appreciated as it deserves.
The absolute equality is introduced in~\cite{FuYuxi} with three points of view.
The first is that bisimulation is a defining property for the internal actions, it is a derived property for the external actions.
This is very much the motivation for the barbed bisimulation.
The second is that a nondeterministic internal action is very different from a deterministic one.
The former must be strongly bisimulated, whereas the latter can be weakly bisimulated.
This is in our opinion the philosophy of the branching bisimulation.
The third is that codivergence and bisimulation complement each other.
Deterministic internal actions are ignorable locally but not globally.
The notion of codivergence is formulated by Priese~\cite{Priese1978}.
The correlation of the codivergence property to the bisimulation property is emphasized independently by Fu~\cite{FuYuxi} and by van Glabbeek, Luttik and Trcka~\cite{vanGlabbeekLuttikTrcka2009}.

Several other bisimulation based equivalences have been proposed for
the $\pi$-calculus. First of all let's take a look at the first
equivalence proposed for the $\pi$-calculus~\cite{MilnerParrowWalker1992}, the early equivalence.
For the sake of illustration let's denote by $\dot{\approx}$ the
largest equipollent codivergent weak bisimulation. We say that $P,Q$
are early equivalent, notation $P\approx_{e}Q$, if
$(\widetilde{c})(P\,|\,O)\dot{\approx}(\widetilde{c})(Q\,|\,O)$ for
every environment $(\widetilde{c})(\_\,|\,O)$. Despite the result
in~\cite{Sangiorgi1992}, it is still an open problem if
$\approx_{e}$ coincides with the weak bisimilarity $\simeq $. But
the issue is not that important in view of the following
incompatibility one finds in the definition of $\approx_{e}$.
\begin{itemize}
\item The bisimulation property together with the equipollence property assume that the environments are dynamically changing.

\item The closure under the environments in the beginning of the observation says the opposite.
\end{itemize}
The so-called late equivalence~\cite{MilnerParrowWalker1992,ParrowSangiorgi1995} has
also been studied in literature, especially when one constructs
proof systems. The late equivalence is not really
an observational equivalence. As long as the one step interactions
are atomic actions, there is no way to tell apart by a third party
the difference between $a(x).P+a(x).Q$ and
$a(x).P+a(x).Q+a(x).([x{=}d]P+[x{\ne}d]Q)$. Even if the late
equivalence is useful when names are treated uniformly, it is a
less attractive equivalence when name dichotomy is applied.

Another well known equivalence for the $\pi$-calculus is the open
bisimilarity~\cite{Sangiorgi1996AI}. The open approach was meant to
deal with the open $\pi$-terms head-on. A closely related
equivalence is the quasi open
bisimilarity~\cite{SangiorgiWalker2001,FuYuxi2005}. An open
bisimulation is a distinction indexed set of relations. In the
presence of the name dichotomy, an indexed family of relations is no
longer necessary. The following is the standard definition of the
open bisimilarity iterated in the present framework.
\begin{quote}
A codivergent bisimulation $\mathcal{R}$ on $\mathcal{T}$
is an {\em open bisimulation} if the following statements are valid.
\begin{enumerate}
\item If $S\mathcal{R}T$ then $S\sigma\;\mathcal{R}\;T\sigma$ for every substitution $\sigma$.
\item If $T\mathcal{R}^{-1}S \stackrel{\ell}{\longrightarrow}S'$ then
$T\stackrel{\ell}{\Longrightarrow}T'\mathcal{R}^{-1}S'$ for some
$T'$.
\item If $S\mathcal{R}T \stackrel{\ell}{\longrightarrow}T'$ then
$S\stackrel{\ell}{\Longrightarrow}S'\mathcal{R}T'$ for some $S'$.
\end{enumerate}
The {\em open bisimilarity} $\approx_o$ is the largest open
bisimulation.
\end{quote}
The relation $\approx_o$ should be given up. For one
thing it fails (\ref{open-eq-1}). The term $\overline{a}a$ for example cannot be simulated by
$[x{=}y]\overline{a}a\,|\,[x{\not=}y]\overline{a}a$. If one tries to correct the above definition, one soon realizes that what one gets is the weak equality.

\subsubsection{Algebraic Property}\label{sec-Algebraic-Property}

This section takes a closer look at the absolute equality on
the $\pi$-terms. The goal is to reduce the proof of an equality between
$\pi$-terms to the proof of an equality between $\pi$-processes. To
describe the reductions, we need to fix some terminology and
notation. We will write $\subseteq_{f}$ for the finite subset
relationship. If
$\mathcal{F}\subseteq_{f}\mathcal{N}\cup\mathcal{N}_{v}$, then
$\mathcal{F}_{c}$ stands for $\mathcal{F}\cap\mathcal{N}$ and
$\mathcal{F}_{v}$ for $\mathcal{F}\cap\mathcal{N}_{v}$. Given such a
finite set $\mathcal{F}$, a condition $\psi$ may completely
characterize the relationships between the elements of
$\mathcal{F}$. This intuition is formalized in the next definition.

\begin{definition}\label{2009-01-29}
Suppose $\mathcal{F}\subseteq_{f}\mathcal{N}\cup\mathcal{N}_{v}$. We
say that a satisfiable condition $\psi$ is {\em complete on}
$\mathcal{F}$, if $\mathcal{F}_{c}=n(\psi)$,
$\mathcal{F}_{v}=v(\psi)$ and for every $x\in\mathcal{F}_v$ and
every $q \in \mathcal{F}$ it holds that either $\psi \Rightarrow x =
q$ or $\psi \Rightarrow x \ne q$. A condition $\psi$ is {\em
complete} if it is complete on $n(\psi)\cup v(\psi)$.
\end{definition}

The defining property of Definition~\ref{2009-01-29} immediately
implies the following lemma.

\begin{lemma}\label{2009-11-13}
Suppose $\varphi$ is complete on $\mathcal{F}$ and $n(\psi)\cup
v(\psi)\subseteq \mathcal{F}$. Then either $\varphi\Rightarrow\psi$
or $\varphi\psi\Rightarrow\bot$.
\end{lemma}

It follows from Lemma~\ref{2009-11-13} that if both $\varphi,\psi$
are complete on $\mathcal{F}$ then either
$\varphi\Leftrightarrow\psi$ or $\varphi\psi\Leftrightarrow\bot$.
This fact indicates that we may apply complete conditions to carry
out case analysis when reasoning about term equality. A case
analysis is based on a complete disjoint partition.

\begin{definition}\label{2009-02-01}
Suppose $\mathcal{F}\subseteq_{f}\mathcal{N}\cup\mathcal{N}_{v}$. A
finite set $\{\varphi_{i}\}_{i\in I}$ is a {\em complete disjoint
partition} of $\mathcal{F}$ if the following conditions are met.
\begin{enumerate}
\item $\varphi_{i}$ is complete on $\mathcal{F}$ for every $i\in I$.
\item $\bigvee_{i\in I}\varphi_{i}\Leftrightarrow\top$.
\item $\varphi_{i}\wedge\varphi_{j}\Rightarrow\bot$ for all $i,j$
such that $i\not=j$.
\end{enumerate}
\end{definition}

Given a set
$\mathcal{F}\subseteq_{f}\mathcal{N}\cup\mathcal{N}_{v}$, there is
always a complete disjoint partition on $\mathcal{F}$. So one can
always apply a partition to a $\pi$-term.

\begin{lemma}\label{2009-05-27}
$T\sim\sum_{i\in I}\varphi_{i}T$, where $\{\varphi_{i}\}_{i\in I}$
is a complete disjoint partition of $gn(T)\cup fv(T)$.
\end{lemma}

Now suppose $\varphi$ is complete on $gn(S\,|\,T)\cup fv(S\,|\,T)$.
Usually it is much easier to prove $\varphi S= \varphi T$ than
to prove $S= T$. So there is a strong interest in the result
stated next.

\begin{lemma}\label{2009-05-28}
Let $\{\varphi_{i}\}_{i\in I}$ be a complete disjoint partition of
$gn(S\,|\,T)\cup fv(S\,|\,T)$. Then $S= T$ if and only if
$\varphi_{i}S= \varphi_{i}T$ for every $i\in I$.
\end{lemma}

How can we reduce the equality proof of $\varphi S= \varphi T$
further if $\varphi$ is complete on say $\{a,b,x,y,z\}=gn(S\,|\,T)\cup fv(S\,|\,T)$?
Suppose $\varphi$ is $[x{=}y][y{=}a][z{=}b]$.
Consider the equality $\varphi S= \varphi T$.
This equality is the same as
\[\varphi S\{a/x,a/y,b/z\} = \varphi T\{a/x,a/y,b/z\},\]
which can be reduced to
$S\{a/x,a/y,b/z\} = T\{a/x,a/y,b/z\}$.
The substitution $\{a/x,a/y,b/z\}$ has the property that it agrees
with the condition $x{=}a\wedge y{=}a\wedge z{=}b$ and its domain
set and range set are subsets of $\{a,b,x,y,z\}$. Substitutions
of this kind are very useful when rewriting $\pi$-terms.

\begin{definition}
An assignment $\rho$ {\em agrees with} $\psi$, and $\psi$ agrees
with $\rho$, if the following statements are valid.
  \begin{enumerate}
  \item $v(\psi)\subseteq dom(\rho)$.
  \item For each $x\in v(\psi)$ and each $a\in n(\psi)$, $\psi\Rightarrow x=a$ if and only if
  $\rho(x)=a$.
  \item For all $x,y\in v(\psi)$, $\psi \Rightarrow x = y$ if and only if $\rho(x) = \rho(y)$.
  \end{enumerate}
A substitution $\sigma$  is {\em induced} by $\psi$ if it agrees with $\psi$ such that
$n(\sigma) \subseteq n(\psi)$, $v(\sigma) \subseteq v(\psi)$ and $dom(\sigma)=v(\psi)$.
\end{definition}
There could be many assignments that agree with $\psi$ and many substitutions that are induced by $\psi$.
We shall write $\rho_{\psi}$ for some assignment that agrees with $\psi$ and
$\sigma_{\psi}$ for some substitution that is induced by $\psi$. It
should be clear that if $\psi$ is complete, then $\sigma_{\psi}(x)
\not= \sigma_{\psi}(y)$ if and only if $\psi \Rightarrow x \not= y$
if and only if $\rho_{\psi}(x) \not= \rho_{\psi}(y)$. The validity
of the next lemma is obvious.

\begin{lemma}
$\varphi T \sim \varphi T\sigma_{\varphi}$.
\end{lemma}
So we have reduced the proof of $S= T$ to the proof of $\varphi
S\sigma_{\varphi}=  \varphi T\sigma_{\varphi}$ for some
$\varphi$ complete on $gn(S\,|\,T)\cup fv(S\,|\,T)$. Now the match
conditions which appear in $\varphi$ are inessential. We may apply the
following lemma to remove all the match conditions.

\begin{lemma}\label{airport}
The following statements are valid.
\begin{enumerate}
\item If $n\notin gn(S\,|\,T)\cup fv(S\,|\,T)$, then
$[x{\not=}n]S= [x{\not=}n]T$ if and only if $S= T$.
\item If $x\notin fv(S\,|\,T)$, then $[x{=}n]S= [x{=}n]T$ if and
only if $S= T$ if and only if $[x{\not=}n]S=
[x{\not=}n]T$.
\end{enumerate}
\end{lemma}

After Lemma~\ref{airport} we may assume that an equality to be
proved is of the form $\delta_{\mathcal{F}}S=
\delta_{\mathcal{F}}T$, where $\mathcal{F}=gn(S\,|\,T)\cup
fv(S\,|\,T)$ and $\delta_{\mathcal{F}}$ is the conjunction
\[\bigwedge\{x\ne n \mid x\in\mathcal{F}_{v},\ n\in \mathcal{F}\ \mathrm{and}\ x,n\ \mathrm{are}\ \mathrm{distinct}\}.\]
In the last step of the reduction we apply the following result.

\begin{lemma}\label{term2process}
$\delta_{\mathcal{F}}S= \delta_{\mathcal{F}}T$ if and only if
$S\rho_{\delta_{\mathcal{F}}}= T\rho_{\delta_{\mathcal{F}}}$,
where $\mathcal{F}=gn(S\,|\,T)\cup fv(S\,|\,T)$.
\end{lemma}

It should be remarked that $\delta_{\mathcal{F}}S=
\delta_{\mathcal{F}}T$ is an equality between two open $\pi$-terms,
whereas $S\rho_{\delta_{\mathcal{F}}}= T\rho_{\delta_{\mathcal{F}}}$ is an equality between two $\pi$-processes.
The proof of Lemma~\ref{term2process} is immediate in view of the following fact.

\begin{lemma}\label{eq-stable}
If $P= Q$ then $P\alpha= Q\alpha$ for every renaming
$\alpha$.
\end{lemma}

We conclude that all equality proofs for $\pi$-terms can be
transformed to equality proofs for $\pi$-processes.

Lemma~\ref{2009-05-28}, Lemma~\ref{airport},
Lemma~\ref{term2process} and Lemma~\ref{eq-stable} are also valid
for $\sim$.

It is difficult to attribute the above lemmas to the original authors.
It would be fair to say that the early researchers on the $\pi$-calculus were all aware of the properties described in these lemmas.
See for example the work of Parrow and Sangiorgi~\cite{ParrowSangiorgi1995} and Lin~\cite{Lin2003}.

\subsection{Testing Equivalence}

The testing equivalence of De Nicola and
Hennessy~\cite{DeNicolaHennessy1984,Hennessy1988} is probably the
best known nonbisimulation equivalence for processes. In the testing
approach, one only cares about the {\em result} of observation, not
the {\em course} of observation. To define a testing equivalence,
one needs to explain what the observers are, how the observers carry
out tests, and what counts as the result of a test. In the spirit of
observation theory, there is not much room for variation. The
observers are the environments. A test by the observer
$(\widetilde{c})(\_\,|\,O)$ on a process $P$ is a {\em complete} internal action
sequence of $(\widetilde{c})(P\,|\,O)$. The result of a test can be
either successful or unsuccessful. From the point of view of
observation, the success of a test can only be indicated by an
observable action.
The only variations that may arise are concerned with the way the successes are reported.

This subsection serves two purposes.
One is to demonstrate that the testing theory for the $\pi$-calculus is just as simple as that for CCS.
The second is to give a model independent characterization of the testing equivalence.
In Section~\ref{sec-May-and-Must-Equivalence} and Section~\ref{sec-Testing-Equivalence-without-Testing}, we will use the {\em separated choice} to achieve greater observing power.
Syntactically the separated choice terms are given by
\[\sum_{i\in I}n_{i}(x).T_{i}\ \mathrm{and}\ \sum_{i\in I}\overline{n_{i}}m_{i}.T_{i}.\]
The semantics is defined by the following standard rules:
\[\begin{array}{ccc}
\inference{} {\sum_{i\in I}a_{i}(x).T_{i}\stackrel{a_{i}c}{\longrightarrow}T_{i}\{c/x\}} \ i\in I \ \ \ &
\inference{} {\sum_{i\in I}\overline{a_{i}}c_{i}.T_{i}\stackrel{\overline{a}c_{i}}{\longrightarrow}T_{i}} \ i\in I
\end{array}\]
The separated choice operator allows one to define derived terms like $\tau.S+a(x).T$.

\subsubsection{May Equivalence and Must Equivalence}\label{sec-May-and-Must-Equivalence}

The testing machinery of De Nicola and Hennessy~\cite{DeNicolaHennessy1984} can be summarized as follows.
\begin{itemize}
\item Success is indicated by the special action $\omega$.

\item In the presence of $\omega$,
the observers can be simplified to the environments of the form
$\_\,|\,O$. There is no need for the localization operator. An observer $O$ is obtained from some process by replacing some occurrences of ${\bf 0}$ by $\omega$.

\item A test of $P$ by $O$ is a complete internal action sequence of $P\,|\,O$.

\item A test of $P$ by $O$ is {\em DH}-{\em successful} if
    at some state of the test, the $\omega$ action is {\em immediately} fireable.
    It is unsuccessful otherwise.

\item A binary relation $\mathcal{R}$ on processes satisfies the {\em may predicate}
if $P\mathcal{R}Q$ implies that, for every observer $O$, some test
of $P\,|\,O$ is DH-successful if and only if some test of $Q\,|\,O$
is DH-successful.

\item A binary relation $\mathcal{R}$ on processes satisfies the {\em must predicate}
if $P\mathcal{R}Q$ implies that, for every observer $O$, all tests
of $P\,|\,O$ are DH-successful if and only if all tests of $Q\,|\,O$
are DH-successful.
\end{itemize}
Using the above terminology, we may define the may equivalence
$\approx_{may}$ as the largest relation on the $\pi$-processes that satisfies the
may predicate, and and the must equivalence $\approx_{must}$,
as the largest relation that satisfies the
the must predicate. Let's see some
illustrating examples.
\begin{itemize}
\item $A\,|\,!\tau\approx_{may}A$. The may equivalence ignores
divergence.

\item $A\,|\,!\tau\not\approx_{must}A$ if $A$ is terminating. The must equivalence is
discriminating against divergent processes.

\item $A\,|\,!\tau\approx_{must}B\,|\,!\tau$ even if
$A\not\approx_{must}B$. It is often said that divergence is
catastrophic for must equivalence.
\end{itemize}
It is obvious from these examples that $\approx_{must}$ is
incompatible to both $\approx_{may}$ and Milner's weak bisimilarity,
which is really undesirable. The may equivalence behaves better. It
is well known that $\approx_{may}$ is nothing but the so called
trace equivalence. A proof of this coincidence can be found in~\cite{DeNicolaHennessy1984}.
\begin{lemma}\label{2009-06-22}
$P\approx_{may}Q$ if and only if
$\;\forall\ell^{*}{\in}\mathcal{L}^{*}.\;(P\stackrel{\ell^{*}}{\Longrightarrow})\Leftrightarrow(Q\stackrel{\ell^{*}}{\Longrightarrow})$.
\end{lemma}
All observational equivalences should subsume the trace equivalence~\cite{vanGlabbeek1990}. For the must equivalence this is
true only for a set of hereditarily terminating processes. The
following definition and lemma are from~\cite{DeNicolaHennessy1984}.
\begin{definition}
A process is {\em strongly convergent} if all the descendants of the
process are terminating.
\end{definition}

\begin{figure*}[t]
\begin{center}
\begin{tabular}{|rclr|}
\hline
$O_{\epsilon}^{\mathcal{F}}$ &$\stackrel{\rm def}{=}$& $\tau.{\bf 0}$ &  \\
$O_{ab,\ell^{*}}^{\mathcal{F}}$ &$\stackrel{\rm def}{=}$& $\tau.\omega+\overline{a}b.O_{\ell^{*}}^{\mathcal{F}}$ &  \\
$O_{\overline{a}b,\ell^{*}}^{\mathcal{F}}$ &$\stackrel{\rm def}{=}$&
$\tau.\omega+a(z).([z{\ne}b]\tau.\omega\,|\,[z{=}b]\tau.O_{\ell^{*}}^{\mathcal{F}\cup\{b\}})$ &  \\
$O_{\overline{a}(b),\ell^{*}}^{\mathcal{F}}$ &$\stackrel{\rm
def}{=}$&
$\tau.\omega+a(z).([z{\in}\mathcal{F}]\tau.\omega\,|\,[z{\notin}\mathcal{F}]\tau.O_{\ell^{*}}^{\mathcal{F}\cup\{z\}})$ &  \\
\hline
\end{tabular}
\end{center}
\caption{Trace Observer. \label{Trace-Observer}}
\end{figure*}

\begin{lemma}\label{mustinmay}
$P\approx_{must}Q$ implies $P\approx_{may}Q$ if $P,Q$ are strongly
convergent.
\end{lemma}
\begin{proof}
We start by defining the {\em trace observer} generated by some
$\ell^{*}\in\mathcal{L}^{*}$ and some finite subset $\mathcal{F}$ of
$\mathcal{N}$. The structural definition is given in
Fig.~\ref{Trace-Observer}. Notice that the input choice and the
output choice are necessary to define the trace observer. Without
loss of generality assume that $P \not\approx_{may} Q$. Then there
must exist some nonempty $\ell^{*}_{1} \in \mathcal{L}^{*}$ such
that $P \stackrel{\ell^{*}_{1}}{\Longrightarrow}$ and $\neg(Q
\stackrel{\ell^{*}_{1}}{\Longrightarrow})$. Let $\mathcal{F}'$ be
$gn(P\,|\,Q)$. Obviously $P\,|\,O_{\ell^{*}_{1}}^{\mathcal{F}'}$ has
an unsuccessful test. Since $Q$ is strongly convergent, it does not
induce any infinite computation. It follows that all the tests of
$Q\,|\,O_{\ell^{*}_{1}}^{\mathcal{F}'}$ are successful.
\end{proof}

\subsubsection{Remark}\label{sec-Notions-of-Success}

The proceeding section has pinpointed the problems with the must
equivalence. Let's summarize.
\begin{enumerate}
\item [I.] Something must be wrong with the fact that Milner's weak bisimilarity
is not a sub-relation of the testing equivalence. Both the weak
bisimilarity and the testing equivalence are to blame since neither
deals with divergence properly. Moreover there is no excuse for
$\approx_{must}\not\subseteq\approx_{may}$.

\item [II.] It is a bit odd to introduce a special symbol $\omega$
to indicate success. The action $\omega$ introduces an asymmetry
between testers and testees. In De Nicola and Hennessy's approach, a
testing equivalence is always defined for a particular process
calculus. The definition may vary from one calculus to another. If
the testing equivalence is really a fundamental equivalence for
processes, there ought to be a model independent definition of the
equivalence that applies to all calculi. The model independent
characterization should be able to provide some canonicity for the
testing approach.
\end{enumerate}
There have been several attempts~\cite{Phillips1987,BrinksmaRensinkVogler1995,NatarajanCleaveland1995,BorealeDeNicolaPugliese1999,BorealeDeNicolaPugliese2001} to modify the definition of must equivalence in order to resolve issues I and II,
Let's review some of the proposals.
\begin{enumerate}
\item [I.]
Brinksma, Rensink and Volger's {\em should equivalence}~\cite{BrinksmaRensinkVogler1995}
and Natarajan and Cleaveland's {\em fair testing
equivalence}~\cite{NatarajanCleaveland1995} are two modifications proposed to address
issue I. These two variants are essentially defined over the same
success predicate. A test of $P$ by $O$ is {\em FS}-{\em successful}
if $Q\stackrel{\omega}{\Longrightarrow}$ whenever
$P\,|\,O\Longrightarrow Q$. A binary relation $\mathcal{R}$ on
processes satisfies {\em fair/should predicate} if $P\mathcal{R}Q$
implies that, for every observer $O$, all tests of $P\,|\,O$ are
FS-successful if and only if all tests of $Q\,|\,O$ are
FS-successful. Let $\approx_{FS}$ be the largest relation that
satisfies this predicate. This equivalence stays between Milner's
weak bisimilarity and the trace equivalence.

\item [II.]
Boreale, De Nicola and Pugliese address issue II
in~\cite{BorealeDeNicolaPugliese1999,BorealeDeNicolaPugliese2001}. Instead of using the special symbol
$\omega$, they let all the observable actions to indicate success.
Their version of the fair/should predicate is based on a slightly
different notion of success. A test of $P$ by $O$ is {\em BDP}-{\em
successful} at $\ell$ if $Q\stackrel{\ell}{\Longrightarrow}$
whenever $P\,|\,O\Longrightarrow Q$. Notice that there is a guarantee at
$\ell$ predicate for every $\ell\in\mathcal{L}$. Boreale, De Nicola
and Pugliese's equivalence $\approx_{BDP}$ is the largest reflexive
contextual relation satisfying all the guarantee at $\ell$
predicates. It is proved in~\cite{BorealeDeNicolaPugliese1999} that $\approx_{BDP}$
coincides with $\approx_{FS}$ for the calculus considered
in~\cite{BorealeDeNicolaPugliese1999}. The significance of their approach is that it
provides a characterization of the must equivalence without
resorting to any testing machinery.
\end{enumerate}
We shall cast more light on $\approx_{FS}$ and $\approx_{BDP}$ next.

\subsubsection{Testing Equivalence without Testing}\label{sec-Testing-Equivalence-without-Testing}

Although the study in~\cite{BorealeDeNicolaPugliese1999} is carried out for a particular
calculus with a particular notion of context, Boreale, De Nicola and
Pugliese's approach to testing equivalence is basically model
independent. In this section we shall give an even more abstract
characterization of $\approx_{BDP}$ in the style of the absolute
equality. We start with a similar abstract characterization of the
trace equivalence. The proof of the following lemma is essentially
given in~\cite{BorealeDeNicolaPugliese1999}.

\begin{lemma}\label{diamond=may}
The equivalence $\approx_{may}$ is the largest reflexive, equipollent, extensional relation.
\end{lemma}

In view of Lemma~\ref{diamond=may} we may introduce the following
definition.

\begin{definition}\label{diamond-eq}
The {\em diamond equality} $=_{\Diamond}$ is the largest reflexive,
equipollent, extensional relation.
\end{definition}

The diamond equality is about the existence of a successful testing.
A logical dual would be about the inevitability of successful
testings. For the purpose of introducing such a dual, one needs to
strengthen the equipollence condition.

\begin{definition}\label{hereditarily-equipollence}
A process $P$ is {\em strongly  observable}, notation
$P{\downdownarrows}$, if $P'{\Downarrow}$ for all $P'$ such that
$P\Longrightarrow P'$. A relation $\mathcal{R}$ is {\em strongly
equipollent} if $P\mathcal{R}Q\Rightarrow(P{\downdownarrows}\Leftrightarrow
Q{\downdownarrows})$.
\end{definition}

After Definition~\ref{hereditarily-equipollence}, the following
definition must be expected.

\begin{definition}\label{box-equality}
The {\em box equality} $=_{\Box}$ is the largest reflexive, strongly
equipollent, extensional relation.
\end{definition}

Without further ado, we summarize the properties of $=_{\Box}$ by the next theorem.
It says that the box equality improves upon the testing equivalence as expected.

\begin{theorem}\label{2008-09-24}
The following statements are valid.
\begin{enumerate}
\item The strict inclusions $=\,\subsetneq\,=_{\Box}\,\subsetneq\,=_{\Diamond}$ hold.
\item The strict inclusion $=_\Box\;\subsetneq\;\approx_{must}$ holds for the strongly convergent $\pi$-processes.
\item The coincidence of $=_{\Box}$ and $\approx_{must}$ holds for the finite $\pi$-processes.
\end{enumerate}
\end{theorem}
\begin{proof}
(1) The inclusion $=\,\subseteq\,=_{\Box}$ is valid by definition.
The strictness of the inclusion is witnessed by the pair $a.b+a.c$, $a.(\tau.b+\tau.c)$.
The proof of Lemma~\ref{mustinmay} can be reiterated to show $=_{\Box}\,\subsetneq\,=_{\Diamond}$.
The strictness is witnessed by the process pair $a.(!b\,|\,!c)$ and $a.b.(!b\,|\,!c)+a.c.(!b\,|\,!c)$.

(2) Suppose  $P \approx_{\Box}Q$ and $P \not\approx_{must}Q$ for strongly convergent processes $P,Q$.
Assume without loss of generality that some observer $O$ exists such that $P\,|\,O$ has an unsuccessful test while all the tests of $Q\,|\,O$ are successful.
There are two cases:
\begin{itemize}
\item
Suppose the failure test of $P\,|\,O$ is the infinite tau action sequence
\begin{equation}\label{2010-04-25}
P\,|\,O \stackrel{\tau}{\longrightarrow}
(\widetilde{c_{1}})(P_{1}\,|\,O_{1}) \ldots
\stackrel{\tau}{\longrightarrow}
(\widetilde{c_{i}})(P_{i}\,|\,O_{i}) \ldots.
\end{equation}
If $O$ performs an infinite {\em consecutive} tau action sequence in (\ref{2010-04-25}), then it must be the case that $P$ has performed only a finite number of interactions with $O$. According to (1), $Q$ must be able to perform the same finite sequence of external actions.
But then $Q\,|\,O$ has an infinite unsuccessful test of the shape (\ref{2010-04-25}), which contradicts the assumption.
If $O$ does not perform any infinite {\em consecutive} tau action sequence in (\ref{2010-04-25}), then the strong convergence property guarantees that $P$ must perform an infinite number of external actions in (\ref{2010-04-25}).
It follows from the assumption $P \not\approx_{must}Q$ that $Q$ cannot perform the same infinite number of external actions.
It follows from the finite branching property that $P\not=_{\diamond}Q$ must hold, contradicting the other assumption.

\item
Suppose the failure test of $P\,|\,O$ is a finite internal action sequence.
Let $\lambda$ be a fresh label and $\widetilde{c}=gn(P\,|\,Q\,|\,O)$.
It is clear that $(\widetilde{c})(Q\,|\,O\{\lambda/\omega\})\downdownarrows$ but not $(\widetilde{c})(P\,|\,O\{\lambda/\omega\})\downdownarrows$.
This is again a contradiction.
\end{itemize}
We conclude that $=_{\Box}\;\subseteq\;\approx_{must}$ holds for the strongly convergent $\pi$-processes.
The strictness of this inclusion is interesting.
Rensink and Vogler~\cite{RensinkVogler2007} point out that
\begin{equation}\label{2010-06-15-Zhang}
 a.\mu X.(a.X+b)+a.\mu X.(a.X+c)
\end{equation}
and
\begin{equation}\label{2010-06-15-Ying}
 a.\mu X.(a.(a.X+c)+b)+a.\mu X.(a.(a.X+b)+c),
\end{equation}
originally from~\cite{BergstraKlopOlderog1987}, are testing equivalent. They are not fair/should equivalent.
In the terminology of this paper (\ref{2010-06-15-Zhang}) and (\ref{2010-06-15-Ying}) are not box equal.

(3) We now prove that $P\approx_{must}Q$ implies $P=_{\Box}Q$ for all finite $\pi$-processes $P,Q$.
Suppose $P\approx_{must}Q$ and $P\not=_\Box Q$ for finite $P$ and $Q$.
Take $\mathcal{F}=gn(P\,|\,Q)$. Without loss of generality, assume that there were $\widetilde{c}$ and $O$ such that
$(\widetilde{c})(P \,|\,O){\downdownarrows}$ holds but $(\widetilde{c})(Q \,|\,O){\downdownarrows}$ is not valid.
Moreover since $P,Q$ are finite, we may assume that $O$ is finite. We may even assume that $O$ does not contain the composition operator.
Let $(\widetilde{c}')(Q'\,|\,O')$ be a descendant of $(\widetilde{c})(Q \,|\,O)$ such that the followings hold:
\begin{itemize}
\item
$(\widetilde{c})(Q\,|\,O)\Longrightarrow(\widetilde{c}')(Q'\,|\,O')\not\Downarrow$ and $(\widetilde{c}')(Q'\,|\,O')$ may not perform any $\tau$-actions.

\item
There is a sequence $\ell^{*}$ such that $Q\stackrel{\ell^{*}}{\Longrightarrow}\stackrel{\lambda}{\longrightarrow}\Longrightarrow Q'$ and $O\stackrel{\overline{\ell^{*}}}{\Longrightarrow}\stackrel{\overline{\lambda}}{\longrightarrow}O_{1}\Longrightarrow O'$.
\end{itemize}
Now construct a new process $\llbracket O\rrbracket$ from $O$.
The effect of the function $\llbracket\_\rrbracket$ on the guarded choice terms can be described as follows:
\begin{eqnarray*}
\llbracket a(x).S+a(x).T\rrbracket &\stackrel{\rm def}{=}& \tau.\omega+a(x).\llbracket S\rrbracket+a(x).\llbracket T\rrbracket, \\
\llbracket \overline{a}m.S+\overline{a}n.T\rrbracket &\stackrel{\rm def}{=}&
\tau.\omega+\overline{a}m.\llbracket S\rrbracket+\overline{a}m.\llbracket T\rrbracket.
\end{eqnarray*}
The function $\llbracket\_\rrbracket$ is structural on the other operators.
The process $\llbracket O\rrbracket$ is further modified as follows:
An observer $O^{\omega}_{\ell^{*}\lambda}$ is obtained from $\llbracket O\rrbracket$ by removing all the components $\tau.\omega$ from the choice operations of $\llbracket O\rrbracket$ that lie in the path, as it were, beginning at $O_{1}$ and ending at $O'$.
It is clear that all tests of $P\,|\,O^{\omega}_{\ell^{*}\lambda}$ are successful, whereas there is at least one unsuccessful test of $Q\,|\,O^{\omega}_{\ell^{*}\lambda}$. This contradicts to the assumption $P\approx_{must}Q$.
\end{proof}

The counter example given by the processes (\ref{2010-06-15-Zhang}) and (\ref{2010-06-15-Ying}) point out that the box equality is fairer than the testing equivalence.
If we assume that the nondeterminism of the choice operator is implemented in a fair or probabilistic manner,
a test of the former by $\mu X.(\overline{a}.X+\overline{b}.\omega)$ may or may not succeed.
But all tests of (\ref{2010-06-15-Ying}) by the tester are successful.
In the opinion of Rensink and Vogler~\cite{RensinkVogler2007}, the fairness assumption is built in the box equality.

So far all the results in testing theory are model dependent. For
example the coincidence between $\approx_{FS}$ and $\approx_{BDP}$
cannot be established in a model independent way. This is because to
make sense of the fair/should testing one has to incorporate the
special symbol $\omega$ into a calculus, which is impossible without
knowing the details of the model. Similarly the coincidence of
$\approx_{FS}$, $\approx_{BDP}$, $\approx_{must}$ on the strongly
convergent processes is also model specific. For a particular model
like the $\pi$-calculus one may show that $\approx_{BDP}$ is the largest
reflexive, strongly  equipollent, extensional relation. But there is
no way to prove the validity of this characterization for all
models. Similarly the coincidence between $\approx_{BDP}$ and
$=_{\Box}$ can only be proved for individual models. The same remark
applies to the trace equivalence. There are labeled transition
systems for which one might not like to talk about traces at all. An
example is a labeled transition semantics for a higher order
calculus that defines transitions from processes to abstractions or
concretions.

It is in the light of the above discussion that the significance of
the diamond equality and the box equality emerge. The former provides {\em
the} universal definition of the trace equivalence, whereas the
latter offers {\em the} universal definition of the testing equivalence.
These definitions apply to all models of interaction. For a
particular model it is possible to give an explicit counterpart of
$=_{\Box}$ or $=_{\Diamond}$. The explicit versions of
$=_{\Diamond}$ and $=_{\Box}$ may depend heavily on the details of
the model. The relationship between $=_{\Box}$ and its explicit
characterization is like that between the absolute equality and the
external bisimilarity.

The outline for a model independent testing theory is now complete.

\subsection{Remark}

The model independent methodology can be applied to analyze dozens
of process equivalences defined in literature. In particular it can
be applied to evaluate the equivalence relations for example in the
linear time branching time spectrum~\cite{vanGlabbeek1990}. We have
confirmed in this paper that the diamond equality is the right
generalization of the trace equivalence, the bottom of the spectrum.
On the top of the spectrum, the absolute equality and the weak equality enjoy the model independent characterizations.
It would be interesting to take a look at the other equivalences of the spectrum from this angle.
Some of the equivalences are studied in the strong case. Their generalizations to the weak case might not be closed under composition.
The model independent approach may help us in searching for the correct definitions.
Good process equivalences are inedpendent of any particular model.

It must be pointed out that it is often difficult to come up with an external characterization of an equality defined in a model independent manner.
Let's take a look at an interesting example.
An explicit characterization of the absolute equality for $\pi^{A}$-calculus is tricky.
It is definitely different from the popular equivalence widely accepted for $\pi^{A}$.
The asynchronous equivalence studied by Honda and Tokoro~\cite{HondaTokoro1991ObjCal}, Boudol~\cite{Boudol1992} and Amadio, Castellani and Sangiorgi~\cite{AmadioCastellaniSangiorgi1996},
satisfies the following equality, where $\simeq_{a}$ stands for the asynchronous bisimilarity.
\begin{equation}\label{eq4asynpi-1}
a(x).\overline{a}x \simeq_{a} {\bf 0}.
\end{equation}
The equality (\ref{eq4asynpi-1}) does not hold for the absolute
equality since it violates the equipollence condition. Another
equality validated by the asynchronous equivalence is
(\ref{eq4asynpi-2}).
\begin{equation}\label{eq4asynpi-2}
!a(x).\overline{a}x \simeq_{a} a(x).\overline{a}x
\end{equation}
Equality (\ref{eq4asynpi-2}) is rejected by the absolute equality
because $\overline{a}c\,|\,!a(x).\overline{a}x$ is divergent but
$\overline{a}c\,|\,a(x).\overline{a}x$ is not. However the following
absolute equality holds.
\begin{equation}\label{eq4asynpi-3}
a(x).\overline{a}x\,|\,a(x).\overline{a}x =_{\pi^{A}}
a(x).\overline{a}x
\end{equation}
How should we reconcile the difference between the absolute equality and the asynchronous equivalence?
A new approach to the asynchronous calculus that gives a satisfactory answer to the question is outlined in~\cite{FuYuxi-Theory-by-Process}.
Asynchrony is more of an application issue than a model theoretical issue.

An important topic that is not discussed in the present paper is the pre-order relations on processes.
A pre-order can be interpreted as an implementation relation or a refinement relation.
Examples of pre-orders are simulation and testing order.
Notice that the box equality immediately suggests the box order $\leq_{\Box}$ relation, which is the model independent counterpart of the testing order.
So far the order theory of processes has been basically model specific.
It should be fruitful to carry out a systematic model independent study on the order relations on the processes.
Initial efforts have been made in~\cite{He2010} along this line of investigation.

We conclude this section by remarking that the branching style
bisimulation property is what makes the coincidence proofs
difficult. External characterizations of the weak equalities of the
$\pi$-variants are much easier to come by.

\section{Expressiveness}\label{sec-Expressiveness}

How do different variants of the $\pi$-calculus compare? One fundamental criterion for comparison is relative expressiveness.
Typical quesitons about expressiveness can be addressed in the following fashion:
Is a calculus $\mathbb{M}$ as powerful as another calculus $\mathbb{L}$?
Or are $\mathbb{L}$ and $\mathbb{M}$ non-comparable in terms of expressiveness?
To answer these questions, we need a notion of relative expressiveness that does not refer to any particular model.
The philosophy developed in~\cite{FuYuxi} is that the relative
expressiveness is the same thing as the process equality. The latter
is a relation on one calculus, whereas the former is a relation from
one calculus to another. To say that $\mathbb{M}$ is at least
as expressive as $\mathbb{L}$ means that for every process $L$ in
$\mathbb{L}$ there is some process $M$ in $\mathbb{M}$ such that $M$
is somehow equal to $L$. It is clear from this statement that
relative expressiveness must break the symmetry of the absolute
equality. Condition (2) of Definition~\ref{absolute-equality} can be
safely kept when comparing two process calculi. But condition (1)
has to be modified. The reader is advised to consult~\cite{FuYuxi}
for the argument why the reflexivity condition for the absolute
equality turns into a totality condition and a soundness condition. The
following definition is taken from~\cite{FuYuxi}.

\begin{definition}\label{subbisimilarity}
Suppose $\mathbb{L},\mathbb{M}$ are two process calculi. A binary
relation $\Re$ from the set $\mathcal{P}_{\mathbb{L}}$ of
$\mathbb{L}$-processes to the set $\mathcal{P}_{\mathbb{M}}$ of
$\mathbb{M}$-processes is a {\em subbisimilarity} if the following
statements are valid.
\begin{enumerate}
\item $\Re$ is {\em reflexive from} $\mathbb{L}$ {\em to} $\mathbb{M}$ in the sense that the following properties hold.
\begin{enumerate}
\item $\Re$ is {\em total}, meaning that
$\forall L\in\mathcal{P}_{\mathbb{L}}.\exists
M\in\mathcal{P}_{\mathbb{M}}.L\Re M$.

\item $\Re$ is {\em sound}, meaning that
$M_{1}\Re^{-1}L_{1}=_{\mathbb{L}}L_{2}\Re M_{2}$ implies
$M_{1}=_{\mathbb{M}}M_{2}$.
\end{enumerate}
\item $\Re$ is equipollent, extensional, codivergent and
bisimilar.
\end{enumerate}
We say that $\mathbb{L}$ is {\em subbisimilar to} $\mathbb{M}$,
notation $\mathbb{L}\sqsubseteq\mathbb{M}$, if there is a
subbisimilarity from $\mathbb{L}$ to $\mathbb{M}$.
\end{definition}

An elementary requirement for the relative expressiveness relationship is transitivity.
The subbisimilarity relationship is well defined in this aspect since its transitivity is apparent from the definition.
So we have formalized the intuitive notion ``$\mathbb{M}$ being at
least as expressive as $\mathbb{L}$'' by
``$\mathbb{L}\sqsubseteq\mathbb{M}$''. We write
$\mathbb{L}\sqsubset\mathbb{M}$ if $\mathbb{L}\sqsubseteq\mathbb{M}$
and $\mathbb{M}\not\sqsubseteq\mathbb{L}$, meaning that $\mathbb{M}$
is strictly more expressive than $\mathbb{L}$. The reader might
wonder why the soundness condition of
Definition~\ref{subbisimilarity} is not strengthened to the full
abstraction. The truth is that the soundness condition is equivalent
to the full abstraction condition as far as
Definition~\ref{subbisimilarity} is concerned.

It is shown in~\cite{FuYuxi} that in general there are an infinite number of pairwise incompatible subbisimilarities from $\mathbb{L}$ to $\mathbb{M}$.
However for the $\pi$-variants studied in~\cite{FuLu2010} there is essentially a unique subbisimilarity between any two of them.
It is shown in~\cite{FuYuxi} that a subbisimilarity for these variants is basically a total relation $\Re$ satisfying the following external characterization property:
\begin{enumerate}
\item  If $Q\Re^{-1}P \stackrel{\ell}{\longrightarrow}P'$ then
$Q\Longrightarrow Q''\stackrel{\ell}{\longrightarrow}Q'\Re^{-1}P'$
and $P\Re Q''$ for some $Q',Q''$.
\item  If $P\Re Q\stackrel{\ell}{\longrightarrow}Q'$ then $P\Longrightarrow
P''\stackrel{\ell}{\longrightarrow}P'\Re Q'$ and $P''\Re Q$ for some
$P',P''$.
\end{enumerate}

Now consider the self translation $\llbracket\_\rrbracket_{\bowtie}$ from the $\pi$-calculus to itself.
The nontrivial part of the translation is defined as follows:
\begin{eqnarray*}
\llbracket a(x).T\rrbracket_{\bowtie} &\stackrel{\rm def}{=}&
\overline{a}(c).c(x).\llbracket T\rrbracket_{\bowtie}, \\
\llbracket \overline{a}b.T\rrbracket_{\bowtie} &\stackrel{\rm
def}{=}& a(y).(\overline{y}b\,|\,\llbracket T\rrbracket_{\bowtie}).
\end{eqnarray*}
It is structural on the non-prefix terms. Is
$\llbracket\_\rrbracket_{\bowtie};=$ a subbisimilarity? The negative
answer is proved in~\cite{FuYuxi}. This is an example of a typical
phenomenon. Without the soundness condition, many liberal encodings
like the above one would be admitted. In most cases soundness is the
only thing to prove when comparing a syntactical subcalculus against a super calculus.

If the external characterizations of the absolute equality in $\pi_{1}$ and $\pi_{2}$ are available, then it is often easy to see if $\pi_{1}\sqsubseteq\pi_{2}$.
Otherwise it might turn out difficult to prove $\pi_{1}\sqsubseteq\pi_{2}$.
For example so far we have not been able to confirm that $\pi^{M}\sqsubseteq\pi^{-}\sqsubseteq\pi$, although the following negative results are easy to derive.

\begin{proposition}\label{2010-01-02}
$\pi\not\sqsubseteq\pi^{-}\not\sqsubseteq\pi^{M}$.
\end{proposition}
\begin{proof}
Using the same idea from~\cite{FuYuxi} it is routine to show that
the $\pi$-process $a(x).\overline{b}b+a(x).\overline{c}c$ cannot be
defined in $\pi^{-}$ and the $\pi^{-}$-process
$a(x).[x{=}c]\overline{b}b$ cannot be defined in $\pi^{M}$.
\end{proof}

\begin{figure}[t]
\begin{center}
\begin{tabular}{|c||c|c|c|c|c|}
\hline
 $\sqsubseteq$ & $\pi$ & $\pi^{L}$ & $\pi^{R}$ & $\pi^{S}$  \\
\hline\hline
$\pi$ & $\surd$ & $\times$ & $\times$ & $?$ \\
\hline
$\pi^{L}$ & $\times$ & $\surd$ & $\times$ & $?$ \\
\hline
$\pi^{R}$ & $\times$ & $\times$ & $\surd$ & $?$ \\
\hline
$\pi^{S}$ & $?$ & $?$ & $?$ & $\surd$ \\
\hline
\end{tabular}
\end{center}
\caption{Expressiveness Relationship
\label{Expressiveness-Relationship}}
\end{figure}

The relative expressiveness between $\pi,\pi^{L},\pi^{R}$ is settled in~\cite{XueLongFu2011}.
Fig.~\ref{Expressiveness-Relationship} shows that they are pairwise incompatible.
The relative expressiveness of $\pi^{S}$ is still unknown.

Before ending this section, let's see a positive result. Consider the boolean expressions constructed from the binary relation $=$ and the logical operators $\wedge,\neg$. For example $\neg(\neg(x=a)\wedge\neg(x=b))$ is such an expression, which is normally abbreviated to $x=a\vee x=b$. Let $\pi^{\vee}$ denote the $\pi$-calculus with the match and mismatch operators replaced by the operator $[\varphi](\_)$, where $\varphi$ is a boolean expression.

\begin{proposition}\label{2010-03-01}
$\pi\sqsubseteq\pi^{\vee}\sqsubseteq\pi$.
\end{proposition}
\begin{proof}
The external bisimilarity and the absolute equality coincide for $\pi^{\vee}$.
So $\pi\sqsubseteq\pi^{\vee}$ is obvious. The proof of the converse is equally simple as soon as the encodings of the processes of the form $[\varphi]T$ become clear.
Given any $\varphi$, one may transform it into a disjunctive normal form $\bigvee_{i\in I}\varphi_{i}$.
By applying the equivalence
\[\bigvee_{i\in I}\varphi_{i} \Leftrightarrow \bigvee_{i\in I}\varphi_{i}\wedge(u=v) \vee \bigvee_{i\in I}\varphi_{i}\wedge(u\not=v),\]
one gets a complete disjoint partition $\bigvee_{j\in J}\psi_{j}$ of $\varphi$ on $fv([\varphi]T)\cup gn([\varphi]T)$.
It follows that
$[\varphi]T = \prod_{j\in J}[\psi_{j}]T$
holds in $\pi^{\vee}$. Therefore $[\varphi]T$ can be bisimulated by $\prod_{j\in J}[\psi_{j}]T$.
\end{proof}

\section{Proof System}\label{sec-Proof-System}

An important issue for a process calculus is to design algorithms
for the decidable fragments of the calculus. The behavior of a
finite term can obviously be examined by a terminating
procedure~\cite{Milner1989}. More generally if a term has only
a finite number of descendants and is finite branching, then there
is an algorithm that generates its transition graph~\cite{Lin1996},
which can be seen as the abstract representation of the term. An
equivalence checker for the terms now works on the transitions
graphs. To help transform the terms to their underlying transition
graphs, a recursive equational rewrite system would be helpful.

In this section we construct two equational systems of the finite
$\pi$-terms, one for the absolute equality and the other for the box
equality. Our prime objective is to demonstrate how the name
dichotomy simplifies equational reasoning.

\subsection{Normal Form}

To construct an equational system, it is always convenient to introduce a combinator that is capable of describing the nondeterministic choices inherent in concurrent systems.
This is the {\em general choice} operator `$+$'.
The semantics of this operator is defined by the following rules.
\[\begin{array}{cc}
\inference{S\stackrel{\lambda}{\longrightarrow}S'}{S+T\stackrel{\lambda}{\longrightarrow}S'}\
\ \ \ \  &
\inference{T\stackrel{\lambda}{\longrightarrow}T'}{S+T\stackrel{\lambda}{\longrightarrow}T'}
\end{array}\]
The choice operation is both commutative and transitive. We will
write $T_{1}+\ldots+T_{n}$ or $\sum_{i\in I}T_{i}$, called a {\em
summation}, without further comment. Here $T_{i}$ is a summand.
Occasionally this notation is confused with the one defined in
(\ref{choice-1}) or (\ref{choice-2}). When this happens, the
confusions are harmless. Using the unguarded choice operator, one may define for example $[p{\in}\mathcal{F}]T$ by
\[\sum_{m\in\mathcal{F}}[p{=}m]T.\]
In general one could define $(\varphi)T$ for a boolean expression constructed from $=,\wedge,\neg$.
It follows that $\textit{if}\;\varphi\;\textit{then}\;S\;\textit{else}\;T$ is definable using the general choice operator.
The choice operator is a debatable operator
since it destroys the congruence property of the absolute equality.
In most of the axiomatic treatment of this operator an induced
congruence relation is introduced. This is avoided in this paper
since we see the choice operator not as a proper process operator
but as an auxiliary one used in the proof systems.

One criterion for an equational system is if it allows one to reduce
every term to some normal form. The exercise carried out in
Section~\ref{sec-Algebraic-Property} suggests the following
definitions.

\begin{definition}\label{normal-form}
Let $\mathcal{F}$ be $gn(T)\cup fv(T)$. The $\pi$-term $T$ is a {\em
normal form} on $\mathcal{F}$ if it is of the form
\[\sum_{i \in I}\lambda_i.T_i\]
such that for each $i \in I$ one of the followings holds.
  \begin{enumerate}
  \item If $\lambda_i = \tau$ then $T_{i}$ is a normal
  form on $\mathcal{F}$.
  \item If $\lambda_i = \overline{n}m$ then $T_{i}$ is a normal
  form on $\mathcal{F}$.
  \item If $\lambda_i = \overline{n}(c)$ then $T_{i}\equiv [c{\notin}\mathcal{F}]T_{i}^{c}$ for some normal
  form $T_{i}^{c}$ on $\mathcal{F}\cup\{c\}$.
  \item If $\lambda_i = n(x)$ then $T_{i}$ is of the form
\[[x{\notin}\mathcal{F}]T_{i}^{\ne} + \sum_{m{\in}\mathcal{F}}[x{=}m]T_{i}^{m}\]
such that $T_{i}^{\ne}$ is a normal form on $\mathcal{F}\cup\{x\}$
and, for each $m\in\mathcal{F}$, $x\notin fv(T_{i}^{m})$ and
$T_{i}^{m}$ is a normal form on $\mathcal{F}$.
  \end{enumerate}
\end{definition}

\begin{definition}
$T$ is a {\em complete normal form} on $\mathcal{F}$ if it is of the
form $\delta_{\mathcal{F}}T'$ for some normal form $T'$ such that
$gn(T')\cup
fv(T')\subseteq\mathcal{F}\subseteq_{f}\mathcal{N}\cup\mathcal{N}_{v}$.
\end{definition}

\subsection{Axiom for Absolute Equality}\label{sec-Axiom-4-Absolute-Equality}

\begin{figure}[t]
\begin{center}
  \begin{tabular}{|lrcl|}
    \hline
    L1 & $(c) {\bf 0}$ & = & ${\bf 0}$ \\
    L2 & $(c) (d) T$ & = & $(d) (c) T$ \\
    L3 & $(c)\pi.T$ & = & $\pi.(c)T$\ \ \ \ \ \ \ \ \ \ \ \ \ \ \ \ \ \ \ \ \ \ \ \ \ \ \ if $c\notin n(\pi)$ \\
    L4 & $(c)\pi.T$ & = & ${\bf 0}$\ \ \ \ \ \ \ \ \ \ \ \ \ \ \ \ \ \ \ \ \ \ \ \ \ \ \ \ \ \ \ \ $\;$ if $c = subj(\pi)$ \\
    L5 & $(c) \varphi T$ & = & $\varphi(c) T$\ \ \ \ \ \ \ \ \ \ \ \ \ \ \ \ \ \ \ \ \ \ \ \ \ \ $\;$ if $c \notin n(\varphi)$ \\
    L6 & $(c) [x{=}c]T$ & = & ${\bf 0}$  \\
    L7 & $(c) (S + T)$ & = & $(c)S + (c)T$ \\
    \hline
    M1 & $(\top)T$ & = & $T$ \\
    M2 & $(\bot)T$ & = & ${\bf 0}$ \\
    M3 & $\varphi T$ & = & $\psi T$\ \ \ \ \ \ \ \ \ \ \ \ \ \ \ \ \ \ \ \ \ \ \ \ \ \ \ \ \ \ \ if $\varphi \Leftrightarrow \psi$ \\
    M4 & $[x{=}p] T$ & = & $[x{=}p] T\{p/x\}$ \\
    M5 & $[x{\not=}p]\pi.T$ & = & $[x{\not=}p]\pi.[x{\not=}p]T$\ \ \ \ \ \ \ \ \ \ \ \ \ \ \ \ if $x\notin bv(\pi)\not\ni p$ \\
    M6 & $\varphi(S+T)$ & = & $\varphi S+\varphi T$ \\
    \hline
    S1 & $T + {\bf 0}$ & = & $T$ \\
    S2 & $S + T$ & = & $T + S$ \\
    S3 & $R + (S + T)$ & = & $(R + S) + T$ \\
    S4 & $T + T$ & = & $T$ \\
    S5 & $T$ & = & $[x{=}p]T + [x{\ne}p]T$ \\
    S6 & $n(x).S + n(x).T$ & = & $n(x).S +n(x).T + n(x).([x{=}p]S + [x{\ne}p]T)$ \\
    \hline
    C1 & $\lambda.\tau.T$ & = & $\lambda.T$ \\
    C2 & $\tau.(S+T)$ & = & $\tau.(\tau.(S+T)+T)$ \\
    \hline
  \end{tabular}
  \caption{Axioms for Absolute Equality.}
  \label{axiom-4-absolute-equality}
\end{center}
\end{figure}

The equational system for the absolute equality on the finite
$\pi$-terms is given in Fig.~\ref{axiom-4-absolute-equality}. The
axioms C1 and C2 are {\em computation laws}. Notice that C2 implies the following equality
\[\tau.T = \tau.(T+\tau.T).\]
We write $AS\vdash
S=T$ if $S=T$ can be derived from the axioms and the rules in
Fig.~\ref{axiom-4-absolute-equality} and the equivalence and
congruence rules. Some derived axioms are summarized in the next
lemma. The proofs of these derived axioms can be found in for
example~\cite{FuYang2003}.

\begin{lemma}\label{derivedaxiom}
The following propositions are valid.
\begin{enumerate}
\item $AS\vdash T=T+\varphi T$.

\item $AS\vdash \varphi\pi.T=\varphi\pi.\varphi T$.

\item $AS\vdash (c)[x{\not=}c]T=(c)T$ and $(c)[x{=}c]T={\bf 0}$.

\item $AS\vdash \varphi T=\varphi T\sigma_{\varphi}$.

\item $AS\vdash \lambda.\varphi\tau.T=\lambda.\varphi T$.

\item $AS\vdash m(x).([x{\in}\mathcal{F}]T'+[x{\notin}\mathcal{F}]\tau.T) + \sum_{i=1}^{k}m(x).([x{\not=}n_{i}]T_{i}'+[x{=}n_{i}]\tau.T) =
m(x).([x{\in}\mathcal{F}]T'+[x{\notin}\mathcal{F}]\tau.T) + \sum_{i=1}^{k}m(x).([x{\not=}n_{i}]T_{i}'+[x{=}n_{i}]\tau.T)+m(x).T$, where $\mathcal{F}=\{n_{1},\ldots,n_{k}\}$.
\end{enumerate}
\end{lemma}

Using these derived axioms it is easy to prove the following lemma by structural induction.

\begin{lemma}\label{cnf-normalization}
Suppose $T$ is a finite $\pi$-term and $\varphi$ is complete on a
finite set $\mathcal{F}\supseteq gn(T)\cup fv(T)$. Then $AS\vdash
\varphi T=\varphi^{=}\delta T'$ for some complete normal form
$\delta T'$ on $gn(\delta T')\cup fv(\delta T')$.
\end{lemma}
\begin{proof}
To start with, $AS\vdash \varphi
T=\varphi^{=}\varphi^{\not=}T=\varphi^{=}(\varphi^{\not=}T)\sigma_{\varphi^{=}}$.
In the second step observe that within $AS$ we can carry out the
reductions, explained in Section~\ref{sec-Algebraic-Property}, to
$(\varphi^{\not=}T)\sigma_{\varphi^{=}}$. Notice that, for each
$\mathcal{F}\subseteq_{f}\mathcal{N}\cup\mathcal{N}_{v}$, although
$\{x{\notin}\mathcal{F}\}\cup\{x{=}m\}_{m{\in}\mathcal{F}}$ is not a
complete disjoint partition of $\mathcal{F}\cup\{x\}$, the set
$\{(x{\notin}\mathcal{F})\wedge\delta_{\mathcal{F}}\}\cup\{(x{=}m)\wedge\delta_{\mathcal{F}}\}_{m{\in}\mathcal{F}}$
{\em is} a complete disjoint partition of $\mathcal{F}\cup\{x\}$.
\end{proof}

Suppose we are to prove $AS\vdash \varphi S=\varphi T$ where
$\varphi$ is complete on $gn(S\,|\,T)\cup fv(S\,|\,T)$. According to
Lemma~\ref{cnf-normalization} we only need to prove $AS\vdash \delta
S'=\delta T'$ for some complete normal forms $\delta S',\delta T'$.
In the light of this fact the next lemma should play a crucial role
in the completeness proof.

\begin{lemma}\label{lem:promotion}
Suppose $S,T$ are normal forms on the finite set
$\mathcal{F}\supseteq gn(S\,|\,T)\cup fv(S\,|\,T)$. If
$\delta_{\mathcal{F}}S = \delta_{\mathcal{F}}T$ then $AS \vdash
\delta_{\mathcal{F}}\tau.S = \delta_{\mathcal{F}}\tau.T$.
\end{lemma}
\begin{proof}
Assume that $S\equiv\sum_{i\in I}\lambda_{i}.S_{i}$ and
$T\equiv\sum_{j\in J}\lambda_{j}.T_{j}$ and that
$\delta_{\mathcal{F}}S = \delta_{\mathcal{F}}T$ for some
$\mathcal{F}\supseteq gn(S\,|\,T)\cup fv(S\,|\,T)$. Let $\rho$ be an assignment that agrees with $\delta_{\mathcal{F}}$.
 We are going to establish by simultaneous induction on the structure of $S,T$ the properties stated below.
\begin{quote}
(S) If $T\rho\rightarrow^{+}T'\rho\nrightarrow$ then $AS\vdash \delta_{\mathcal{F}}\tau.T = \delta_{\mathcal{F}}\tau.(T+T') = \delta_{\mathcal{F}}\tau.T'$.
\end{quote}
\begin{quote}
(P) If $\delta_{\mathcal{F}}S = \delta_{\mathcal{F}}T$ then $AS \vdash
\delta_{\mathcal{F}}\tau.S = \delta_{\mathcal{F}}\tau.(S+T) = \delta_{\mathcal{F}}\tau.T$.
\end{quote}
Suppose $T\rho\rightarrow T_{1}\rho\rightarrow\ldots\rightarrow T_{n}\rho\rightarrow T'\rho\nrightarrow$.
Lemma~\ref{term2process} and Lemma~\ref{eq-stable} imply that $\delta_{\mathcal{F}}T_{1}=\ldots= \delta_{\mathcal{F}}T_{n}= \delta_{\mathcal{F}}T'$.
By induction hypothesis on (P), $AS\vdash\delta_{\mathcal{F}}\tau.T_{1}=\ldots=\delta_{\mathcal{F}}\tau.T_{n}=\delta_{\mathcal{F}}\tau.T'$.
Therefore
\[
AS\vdash \delta_{\mathcal{F}}\tau.T = \delta_{\mathcal{F}}\tau.(T+\tau.T').
\]
Let $\lambda_{j}.T_{j}$ be a summand of $T$.
Consider how $T'\rho$ might bisimulate $T\rho\stackrel{\lambda_{j}}{\longrightarrow}T_{j}\rho$. There are four cases.
\begin{enumerate}
\item $\lambda_j = \tau$ and $T\rho\stackrel{\tau}{\longrightarrow}T_{j}\rho$.
Suppose the $\tau$-action is bisimulated by
$T'\rho\stackrel{\tau}{\longrightarrow} T_{j'}'\rho
= T_{j}\rho$ for some $T_{j'}'$. Then $\delta_{\mathcal{F}} T_{j'}'
= \delta_{\mathcal{F}} T_{j}$ by Lemma~\ref{term2process} and Lemma~\ref{eq-stable}. It
follows from the induction hypothesis on (P) that $AS \vdash
\delta_{\mathcal{F}}\tau.T_{j'}'=\delta_{\mathcal{F}}\tau.T_{j}$.
Using (2) of Lemma~\ref{derivedaxiom} one gets the following
inference
\begin{eqnarray*}
AS\vdash \delta_{\mathcal{F}}\tau.(T+\tau.T')
 &=& \delta_{\mathcal{F}}\tau.(T+\tau.(T'+\tau.T_{j'}')) \\
 &=& \delta_{\mathcal{F}}\tau.(T+\tau.(T'+\tau.T_{j})).
\end{eqnarray*}
If $T\rho\stackrel{\tau}{\longrightarrow}T_{j}\rho$ is bisimulated vacuously by $T'\rho$.
Then $\delta_{\mathcal{F}}T_{j}=\delta_{\mathcal{F}}T'$ according to Lemma~\ref{term2process} and Lemma~\ref{eq-stable} and consequently $AS\vdash\delta_{\mathcal{F}}\tau.T_{j}=\delta_{\mathcal{F}}\tau.T'$ by the induction hypothesis on (P).
It follows from C2 that
\begin{eqnarray*}
AS\vdash \delta_{\mathcal{F}}\tau.(T+\tau.T')
 &=& \delta_{\mathcal{F}}\tau.(T+\tau.(T'+\tau.T')) \\
 &=& \delta_{\mathcal{F}}\tau.(T+\tau.(T'+\tau.T_{j})).
\end{eqnarray*}

  \item $\lambda_j=\overline{n}m$ and $T\rho\stackrel{\overline{a}b}{\longrightarrow}T_{j}\rho$ for some $a,b$.
  We can prove as in the next case that
  $AS\vdash \delta_{\mathcal{F}}\tau.(T+\tau.T')=\delta_{\mathcal{F}}\tau.(T+\tau.(T'+\overline{n}m.T_{j}))$.

\item $\lambda_j=\overline{m}(c)$ and $T\rho\stackrel{\overline{a}(c)}{\longrightarrow}T_{j}\rho$ for some $a$.
Suppose this action is bisimulated by
$T'\rho\stackrel{\overline{a}(c)}{\longrightarrow}T_{j'}'\rho=
T_{j}\rho$ for some $T_{j'}'$. Then
$\delta_{\mathcal{F}}[c{\notin}\mathcal{F}]T_{j'}'=
\delta_{\mathcal{F}}[c{\notin}\mathcal{F}]T_{j}$ and \[AS \vdash
\delta_{\mathcal{F}}[c{\notin}\mathcal{F}]\tau.T_{j'}' =
\delta_{\mathcal{F}}[c{\notin}\mathcal{F}]\tau.T_{j}\] by induction hypothesis on (P). Therefore \[AS
\vdash \delta_{\mathcal{F}}\overline{m}(c).[c{\notin}\mathcal{F}]T_{j'}'
= \delta_{\mathcal{F}}\overline{m}(c).[c{\notin}\mathcal{F}]T_{j}.\]
It follows from (3) of Lemma~\ref{derivedaxiom} that \[AS\vdash \delta_{\mathcal{F}}\tau.(T+\tau.T')=\delta_{\mathcal{F}}\tau.(T+\tau.(T'+\overline{m}(c).T_{j})).\]

  \item $\lambda_j=m(x)$. Assume that $\rho(m)=a$ and $gn((S\,|\,T)\rho)=\{b_1, \ldots, b_n\}$.
  Let $n_{k}\in \mathcal{F}$ be such that $n_{k}=b_{k}$ if $b_{k}\in\mathcal{F}$
  and $\rho(n_{k})=b_{k}$ if $b_{k}\not\in\mathcal{F}$.
  There are two subcases.
\begin{enumerate}
\item For every $k\in\{1,\ldots,n\}$, one has $T\rho\stackrel{ab_{k}}{\longrightarrow}T_{j}\rho\{b_{k}/x\}$.
Let this action be bisimulated by
$T'\rho\stackrel{ab_{k}}{\longrightarrow}T_{n_{k}}'\rho\{b_{k}/x\} = T_{j}\rho\{b_{k}/x\}$.
Then \[AS\vdash \delta_{\mathcal{F}}[x{=}n_{k}]\tau.T_{n_{k}}'=\delta_{\mathcal{F}}[x{=}n_{k}]\tau.T_{j}\] by induction hypothesis on (P).
Consequently
\begin{eqnarray*}
\delta_{\mathcal{F}}\tau.(T+\tau.T') &=& \delta_{\mathcal{F}}\tau.(T+\tau.(T'+m(x).T_{n_{k}}')) \\
 &=& \delta_{\mathcal{F}}\tau.(T+\tau.(T'+m(x).([x{\not=}n_{k}]\tau.T_{n_{k}}'+[x{=}n_{k}]\tau.T_{n_{k}}'))) \\
 &=& \delta_{\mathcal{F}}\tau.(T+\tau.(T'+m(x).([x{\not=}n_{k}]\tau.T_{n_{k}}'+[x{=}n_{k}]\tau.T_{j}))).
\end{eqnarray*}

\item Suppose $T\rho\stackrel{ad}{\longrightarrow}T_{j}\rho\{d/x\}$, where $d\not\in\mathcal{F}$, is
simulated by
\[T'\rho \stackrel{ad}{\longrightarrow}T_{x}\rho\{d/x\} = T_{j}\rho\{d/x\}.\]
Like in the previous subcase, we may prove that
\[
AS \vdash \delta_{\mathcal{F}}\tau.(T+\tau.T') =
\delta_{\mathcal{F}}\tau.(T+\tau.(T'+m(x).([x{\in}\mathcal{F}]\tau.T_{x}+[x{\notin}\mathcal{F}]\tau.T_{j}))).
\]
\end{enumerate}
Putting together the two equalities obtained in (a) and (b), we
get the following equational rewriting
\begin{eqnarray*}
AS \vdash \delta_{\mathcal{F}}\tau.(T+\tau.T')
&=& \delta_{\mathcal{F}}\tau.(T+\tau.(T' + m(x).([x{\in}\mathcal{F}]\tau.T_{x} +[x{\notin}\mathcal{F}]\tau.T_{j}) \\
 && +\; \sum_{j=1}^{n}m(x).([x{\not=}n_{k}]\tau.T_{n_{k}}+[x{=}n_{k}]\tau.T_{j}))) \\
&=& \delta_{\mathcal{F}}\tau.(T+\tau.(T' + m(x). T_{j})),
\end{eqnarray*}
where the second equality holds by (6) of
Lemma~\ref{derivedaxiom}.
\end{enumerate}
Using the fact that $\delta_{\mathcal{F}}T=\delta_{\mathcal{F}}T'$ we may apply the above argument to every summand of $T$ to derive that
\[
AS \vdash \delta_{\mathcal{F}}\tau.T = \delta_{\mathcal{F}}\tau.(T+\tau.(T'+T)).
\]
Resorting to the full power of C2 we get from the above equality the following:
\begin{equation}\label{2010-04-05-shaoxing}
AS \vdash \delta_{\mathcal{F}}\tau.T = \delta_{\mathcal{F}}\tau.(T+T').
\end{equation}
If we examine the proof of (\ref{2010-04-05-shaoxing}) carefully we realize that it also establishes the following fact:
\[
AS \vdash \delta_{\mathcal{F}}\tau.(T+T') = \delta_{\mathcal{F}}\tau.T'.
\]
This finishes the inductive proof of (S).

The inductive proof of (P) is now simpler.
Suppose $\delta_{\mathcal{F}}S = \delta_{\mathcal{F}}T$. Let $S',T'$ be such that $S\rho\rightarrow^{*}S'\rho\nrightarrow$ and $T\rho\rightarrow^{*}T'\rho\nrightarrow$. According to the induction hypothesis of (S), the followings hold.
\begin{eqnarray*}
AS &\vdash& \delta_{\mathcal{F}}\tau.S = \delta_{\mathcal{F}}\tau.S', \\
AS &\vdash& \delta_{\mathcal{F}}\tau.T = \delta_{\mathcal{F}}\tau.T'.
\end{eqnarray*}
The inductive proof of (S) can be reiterated to show $AS\vdash \delta_{\mathcal{F}}\tau.S'=\delta_{\mathcal{F}}\tau.T'$.
Hence $AS\vdash \delta_{\mathcal{F}}\tau.S=\delta_{\mathcal{F}}\tau.T$.
\end{proof}

It is a small step from Lemma~\ref{lem:promotion} to the
completeness result.

\begin{theorem}\label{ae:completeness}
Suppose $S,T$ are finite $\pi$-terms.
Then $S= T$ if and only if $AS \vdash \tau. S = \tau.T$.
\end{theorem}
\begin{proof}
Let $\mathcal{F}$ be $gn(S\,|\,T)\cup fv(S\,|\,T)$ and let
$\{\varphi_{i}\}_{i\in I}$ be a complete disjoint partition
$\mathcal{F}$. Then $S= T$ if and only if $\varphi_{i}S=
\varphi_{i}T$ for every $i\in I$. For each $i\in I$,
$\varphi_{i}S= \varphi_{i}T$ can be turned into an equality of
the form $\varphi_{i}^{=}(\varphi_{i}^{\not=}
S')\sigma_{\varphi_{i}^{=}}=
\varphi_{i}^{=}(\varphi_{i}^{\not=}T')\sigma_{\varphi_{i}^{=}}$.
Using Lemma~\ref{airport} and Lemma~\ref{cnf-normalization}, this
equality can be simplified to $\delta S''= \delta T''$, where
$\delta S''$ and $\delta T''$ are complete normal forms. We are done
by applying Lemma~\ref{lem:promotion}.
\end{proof}

\subsection{Axiom for Box Equality}\label{sec-Axiom-4-Box-Equality}

According to Theorem~\ref{2008-09-24}, a proof system for the box
equality on the finite $\pi$-terms is the same as a proof system for
the testing equivalence on the finite $\pi$-terms. De Nicola and
Hennessy~\cite{DeNicolaHennessy1984} have constructed an equational system for the testing
equivalence on the finite CCS processes.
Built upon that system, Boreale and De Nicola~\cite{BorealeDeNicola1995} have studied the
equational system for the testing equivalence on the finite
$\pi$-processes. So there is not much
novelty about an equational proof system for the box equality on the
finite $\pi$-terms. It would be however instructive in the present
framework to give an outline of the proof technique that reduces the
completeness for the box equality to the completeness for the
absolute equality.

\begin{figure}[t]
\begin{center}
  \begin{tabular}{|lrclr|}
    \hline
    N1 & $\lambda.S + \lambda.T$ & = & $\lambda.(\tau. S+\tau.T)$ & \\
    N2 & $S + \tau. T$ & = & $\tau. (S + T) + \tau. T$ & \\
    N3 & $n(x). S + \tau. (n(x). T + R)$ & = & $\tau. (n(x). S + n(x). T + R)$ & \\
    N4 & $\overline{n}l. S + \tau. (\overline{n}m. T + R)$ & = & $\tau. (\overline{n}l. S + \overline{n}m. T + R)$ & \\
    \hline
  \end{tabular}
  \caption{Axioms for Box Equivalence.}
  \label{axiom-4-box-equality}
\end{center}
\end{figure}

Since $\simeq \subseteq=_{\Box}$, one may devise an equational
system for the latter by extending the system given in
Fig.~\ref{axiom-4-absolute-equality}. The additional axioms are
given in Fig.~\ref{axiom-4-box-equality}. These laws are the well
known axioms for the testing equivalence
of De Nicola and Hennessy~\cite{DeNicolaHennessy1984} adapted to the $\pi$-calculus. It is
well known that they subsume Milner's $\tau$-laws (See
Fig.~\ref{Milner-law}). Let $AS_{b}$ be the system
$AS\setminus\{C1,C2\}\cup\{N1,N2,N3,N4\}$. It is routine to check
that $AS_{b}$ is sound for the box equality. To prove that the
system is also complete, we apply the following strategy.
\begin{quote}
$S=_{\Box}T$ if and only if there exist some $S',T'$ such that
$AS_{b}\vdash S=S'$, $AS_{b}\vdash T=T'$ and $S'=T'$.
\end{quote}
The soundness of the strategy relies on the fact that $AS_{b}$ is
complete for $=$. To make the strategy work, we need to
introduce a special set of $\pi$-terms, different from the complete
normal forms, so that the box equality and the absolute equality
coincide on these special $\pi$-terms.

\begin{figure}[t]
\begin{center}
  \begin{tabular}{|lrcl|}
    \hline
    T1 & $\lambda.\tau.T$ & = & $\lambda.T$ \\
    T2 & $T+\tau.T$ & = & $\tau.T$ \\
    T3 & $\lambda.(S+\tau.T)$ & = & $\lambda.(S+\tau.T)+\lambda.T$ \\
    \hline
  \end{tabular}
  \caption{Milner's Tau Laws.}
  \label{Milner-law}
\end{center}
\end{figure}

Axiom N1 suggests that a summation may be rewritten to a form in
which no two summands have identical non-tau prefix. For instance
\[AS_{b}\vdash a(x).S+a(y).T=a(z).(\tau.S\{z/x\}+\tau.T\{z/y\})\] and
\[AS_{b}\vdash
\overline{a}(b).S+\overline{a}(c).T=\overline{a}(d).(\tau.S\{d/b\}+\tau.T\{d/c\}).\]
Axiom N2 implies that either all the summands of a summation are
prefixed by $\tau$, or none of them is prefixed by $\tau$. Moreover
N1 actually says that one does not have to consider any
$\tau$-prefix immediately underneath another $\tau$-prefix. Axioms
N3 and N4 can be used to expand a summation $\sum_{i\in
I}\tau.T_{i}$ to a saturated form. We say that $\sum_{i\in
I}\tau.T_{i}$ is {\em saturated} if $\tau.(\lambda_{1}.T_{1}+T_{j})$
is a summand of $\sum_{i\in I}\tau.T_{i}$ whenever there is a
summand $\lambda_{1}.T_{1}$ of $T_{i}$ such that for every summand
$\lambda_{2}.T_{2}$ of $T_{j}$ it holds up to $\alpha$-conversion
that $\lambda_{2}\not=\lambda_{1}$. These observations lead to the
following definition.

\begin{definition}
Let $\mathcal{F}$ be $gn(T)\cup fv(T)$. A finite $\pi$-term $T$ is a
\emph{box normal form} if it is in one of the following two forms:
\begin{enumerate}
\item There are $A,B,C\subseteq_{f}\mathcal{N}$, where $C\subseteq
B$, such that $T$ is of the shape:
\begin{equation}\label{2009-02-04}
\sum_{a\in
A}a(x).\left([x{\not\in}\mathcal{F}]T_{a}+\sum_{n{\in}\mathcal{F}}[x{=}n]T_{a}^{n}\right)
 + \sum_{b\in B}\sum_{n\in N_{b}}\overline{b}n.T_{b}^{n}
 + \sum_{c\in C}\overline{c}(d).T_{c}
\end{equation}
where $N_{b}\subseteq_{f}\mathcal{N}\cup\mathcal{N}_{v}$; and
moreover the following properties hold:
\begin{enumerate}
\item for all $a\in A$, $T_{a}$ is a box normal form on $\mathcal{F}\cup\{x\}$;

\item for all $a\in A$ and all $n{\in}\mathcal{F}$, $x\not\in fv(T_{a}^{n})$ and $T_{a}^{n}$ is a box normal form on $\mathcal{F}$;

\item for all $b\in B$ and all $n{\in}N_{b}$, $T_{b}^{n}$ is a box normal form on $\mathcal{F}$;

\item for all $c\in C$, $T_{c}$ is a box normal form on
$\mathcal{F}\cup\{d\}$.
\end{enumerate}

\item There is some finite set $I$ such that $T$ is of the following
form
\begin{equation}\label{tauhnf}
\sum_{i\in I}\tau.T_{i}
\end{equation}
such that the following properties hold:
\begin{itemize}
\item
For each $i\in I$, $T_{i}$ is a box normal form on $\mathcal{F}$ of
the shape (\ref{2009-02-04});

\item $\sum_{i\in I}\tau.T_{i}$ is saturated.
\end{itemize}
\end{enumerate}
\end{definition}

In the above definition we have ignored the conditionals. This is
rectified in the next definition.

\begin{definition}
$T$ is a {\em complete box normal form} if $T\equiv
\delta_{\mathcal{F}}T'$ for some box normal form $T'$ such that
$gn(T')\cup
fv(T')\subseteq\mathcal{F}\subseteq_{f}\mathcal{N}\cup\mathcal{N}_{v}$.
\end{definition}

We can now state a lemma that correlates
Lemma~\ref{cnf-normalization}.

\begin{lemma}\label{cbnf-normalization}
Suppose $T$ is a finite $\pi$-term and $\varphi$ is complete on a
finite set $\mathcal{F}\supseteq gn(T)\cup fv(T)$. Then
$AS_{b}\vdash \varphi T=\varphi^{=}\delta T'$ for some complete box
normal form $\delta T'$ on $gn(\delta T')\cup fv(\delta T')$.
\end{lemma}
\begin{proof}
The proof is a modification of the proof of
Lemma~\ref{cnf-normalization} with additional rewriting using
N-laws.
\end{proof}

We could have a lemma that parallels Lemma~\ref{lem:promotion}. But
the following result is more revealing.

\begin{lemma}\label{box2absolute}
Suppose $S,T$ are complete box normal forms on the finite set
$\mathcal{F}=gn(S\,|\,T)\cup fv(S\,|\,T)$. Then $S=_{\Box}T$ if and
only if $S=T$.
\end{lemma}
\begin{proof}
Suppose $S=_{\Box}T$. By Lemma~\ref{cbnf-normalization} we may
assume that $S\equiv \delta S'$ and $T\equiv \delta T'$. Let $\rho$
be an assignment that agrees with $\delta$. Then
$S'\rho=_{\Box}T'\rho$. In view of Lemma~\ref{term2process}, we only
need to show that $S'\rho=T'\rho$. So let $\mathcal{R}$ be the
relation
\[\{(P,Q) \mid P=_{\Box}Q,\ \mathrm{and}\ P,Q\ \mathrm{are}\ \mathrm{box}\ \mathrm{normal}\ \mathrm{forms}\}.\]
There are three crucial properties about this relation.
\begin{enumerate}
\item If $P$ is of type~(\ref{tauhnf}) and $Q$ is of
type~(\ref{2009-02-04}), then $P'\mathcal{R}Q$ whenever
$P\stackrel{\tau}{\longrightarrow}P'$. This is so simply because $Q$
cannot do any $\tau$-action.

\item If both $P$ and $Q$ are of type~(\ref{tauhnf}), then
$P\stackrel{\tau}{\longrightarrow}P'$ implies
$Q\stackrel{\tau}{\longrightarrow}Q'\mathcal{R}P'$.

\item If both $P$ and $Q$ are of type~(\ref{2009-02-04}), then
$P\stackrel{\lambda}{\longrightarrow}P'$ implies
$Q\stackrel{\lambda}{\longrightarrow}Q'\mathcal{R}P'$.
\end{enumerate}
For the detailed proofs of these claims, the reader is advised to
consult~\cite{DeNicolaHennessy1984,BorealeDeNicola1995}. So
$\mathcal{R}$ is a $\pi$-bisimulation.
\end{proof}

Theorem~\ref{ae:completeness}, Lemma~\ref{box2absolute} and the fact
that complete box normal forms are complete normal forms immediately
imply the completeness of $AS_{b}$.

\begin{theorem}\label{box:completeness}
Suppose $S,T$ are finite $\pi$-processes. Then $S=_{\Box}T$ if and
only if $AS_{b} \vdash \tau. S = \tau.T$.
\end{theorem}

It is remarkable that axioms N1 through N4 actually reduce the box equality on the finite terms to the absolute equality on the finite terms.
This is yet another support to the branching style bisimulation.

\subsection{Remark}

In theory of CCS, Milner's equational systems for the strong and the
weak congruences on the finite
CCS-processes~\cite{HennessyMilner1985,Milner1989} and De Nicola and
Hennessy's system for the testing congruence are well known.
A complete system for the branching congruence was given van Glabbeek and Weijland~\cite{vanGlabbeekWeijland1989-first-paper-bb} in the first paper on branching bisimulation,
in which the following single tau law is proposed.
\begin{equation}\label{2010-01-14}
\lambda.(\tau.(S+T)+T) = \lambda.(S+T).
\end{equation}
The axiom~(\ref{2010-01-14}) is clearly equivalent to the combination of C1 and C2.
Their proof of the completeness in~\cite{vanGlabbeekWeijland1989-first-paper-bb} makes use of a graph rewriting system.
Later~\cite{vanGlabbeek1993} gave a more traditional proof of the completeness theorem.

The extension of these systems to the value-passing framework is not
trivial, the reason being that every value-passing calculus is built
on top of an oracle domain, say $\mathfrak{D}$. Early inference
systems studied by Hennessy and
Ing\'{o}lfsd\'{o}ttir~\cite{Hennessy1991,HennessyIng1993a,HennessyIng1993b}
are based on concrete semantics. An uncomfortable rule in all these
systems, from the point of view of a proof system, is the so-called
$\omega$-data-rule (\ref{omega-vpc}).
\begin{equation}\label{omega-vpc}
\inference{\forall
v\in\mathfrak{D}.S\{v/x\}=T\{v/x\}}{a(x).S=a(x).T}
\end{equation}
A significant step was made by Hennessy and Lin~\cite{HennessyLin1995} that introduces a whole new approach to
the study of the value-passing calculi. The symbolic semantics
dispenses with (\ref{omega-vpc}) by introducing a strong logic. One
could argue that this use of a logic is cheating because it
essentially makes use of a universal quantification operator that
ranges over the oracle domain $\mathfrak{D}$. But the virtue of the
symbolic approach is that one could define a value-passing calculus
using a moderate logic that makes a lot of sense from a programming
point of view, and then works out the observational theory with the
help of the logic. Using this idea Hennessy and Lin~\cite{HennessyLin1996} propose several symbolic proof systems for finite
value-passing processes. A survey of the symbolic approach is given
in~\cite{IngolfsdottirLin2001}.

For the name-passing calculi, the study on the proof systems was
initiated in the pioneering paper of Milner, Parrow and Walker~\cite{MilnerParrowWalker1992}.
Their system is complete for the
strong early equivalence. It contains the $\omega$-name-rule
(\ref{omega-pi}), which is a variant of (\ref{omega-vpc}).
\begin{equation}\label{omega-pi}
\inference{\forall n\in\mathcal{N}.S\{n/x\}=T\{n/x\}}{a(x).S=a(x).T}
\end{equation}
In the presence of the finite branching property, (\ref{omega-pi})
is more manageable than (\ref{omega-vpc}) since only a finite number of the
premises of (\ref{omega-pi}) have to be verified. A beautiful
alternative to {\em rule} (\ref{omega-pi}) is set of {\em axioms} for match/mismatch introduced by Parrow and Sangiorgi~\cite{ParrowSangiorgi1995},
among which the following one, S5,  plays an indispensable role.
\begin{equation}\label{ps-axiom}
[x{=}y]T+[x{\ne}y]T = T
\end{equation}
Axiom (\ref{ps-axiom}) offers the possibility to carry out case
analysis within an equational system. It is obvious from
(\ref{ps-axiom}) that Parrow and Sangiorgi's systems are only good
for the $\pi$-calculi with the mismatch operator. The mismatch is
also necessary to state the following law, S6, which first appeared
in~\cite{ParrowSangiorgi1995}.
\begin{equation}\label{early-axiom}
n(x).S+n(x).T = n(x).S+n(x).T+n(x).([x{=}y]S+[x{\ne}y]T)
\end{equation}
Axiom (\ref{early-axiom}) characterizes the {\em atomic} nature of
interactions. The symbolic approach to proof systems however makes
use of neither (\ref{omega-pi}) nor (\ref{early-axiom}). Lin's
symbolic proof systems~\cite{Lin1995,Lin2003} are capable of dealing
with calculi with or without mismatch operator. Instead of
(\ref{early-axiom}), the systems in~\cite{Lin1995,Lin2003} resort to
a less attractive rule (\ref{early-rule}).
\begin{equation}\label{early-rule}
\inference{\sum_{i\in I}\tau.S_{i}=\sum_{j\in J}\tau.T_{j}}
{\sum_{i\in I}a(x).S_{i}=\sum_{j\in J}a(x).T_{j}}
\end{equation}
A proof system for the strong open bisimilarity is given
in~\cite{Sangiorgi1996AI}, in which all axioms are indexed by
distinctions. The system without using distinctions is proposed
in~\cite{FuYang2003}. Since the following law, M5,
\begin{equation}\label{m5-axiom}
[x{\ne}y]\pi.T = [x{\ne}y]\pi.[x{\ne}y]T
\end{equation}
is invalid in the open semantics, axiom (\ref{open-axiom}) is
proposed in~\cite{FuYang2003} as a substitute for (\ref{m5-axiom}).
\begin{equation}\label{open-axiom}
(a)C[[x{=}a]T] = (a)C[\mathbf{0}]
\end{equation}
It is worth remarking that (\ref{open-axiom}) is equivalent to
$(a)C[[x{\ne}a]T] = (a)C[T]$ in the presence of (\ref{ps-axiom}). So
it is enough even for the $\pi$-calculus with the mismatch operator.

The first complete proof system for the weak congruence on finite
$\pi$-processes is Lin's symbolic
system~\cite{Lin1995-fix-ind,Lin1998}. The nonsymbolic systems were
proposed by Parrow~\cite{Parrow1999} using (\ref{early-rule}), and
by Fu~\cite{FuYang2003} using (\ref{early-axiom}).
In~\cite{FuYang2003} the authors have also discussed complete
equational systems for the weak open congruence. It is revealed that
the open bisimilarities, as well as the quasi open
bisimilarities~\cite{FuYuxi2005}, are quite complicated in the
presence of the mismatch operator. It has been suggested
that the complications are due to the introduction of the mismatch
operator. The real culprit is however the confusion of the names
and the name variables. Study on the weak open systems pointed out
that Milner's three $\tau$-laws are insufficient~\cite{FuYang2003}.
An additional $\tau$-law
\begin{equation}\label{4thtaulaw}
\tau.T = \tau.(T+\varphi\tau.T)
\end{equation}
is necessary. Notice that (\ref{4thtaulaw}) is derivable from
(\ref{m5-axiom}).

All of these tiny discrepancies appear a little confusing.
What the present paper reveals is that all these incompatibilities disappear once the names are treated properly.
Moreover, since the mismatch operator comes hand in hand with the match operator, there is no real interest in systems without the mismatch operator.
Proof systems for $\pi$-variants can be obtained by extending $AS$ with additional laws.
For example the following two axiom schemes are introduced in~\cite{XueLongFu2011}.
\begin{eqnarray}
\overline{n}(c).C[\overline{c}m.T] & = & \overline{n}(c).C[{\bf 0}], \label{2011-01-26-L} \\
\overline{n}(c).C[c(x).T] & = & \overline{n}(c).C[{\bf 0}]. \label{2011-01-26-R}
\end{eqnarray}
It is pointed out by Xue, Long and Fu~\cite{XueLongFu2011} that $AS\cup\{(\ref{2011-01-26-L})\}$ is complete for $\pi^{L}$ and that $AS\cup\{(\ref{2011-01-26-R})\}$ is complete for $\pi^{R}$.

A more challenging issue is to construct proof systems for the
regular processes~\cite{Milner1984}. A regular process is not
necessarily finite state~\cite{Milner1989}; it is generally finite
control~\cite{Lin1998}. Milner addressed the issue in the framework
of CCS. His strong complete proof system~\cite{Milner1984} and weak
complete proof system~\cite{Milner1989IC} make crucial use of the
following fixpoint induction rule
\begin{equation}\label{fixinduction}
\inference{F\{E/X\}=E}{E=\mu X.F}\ X\ \mathrm{is}\ \mathrm{guarded}\
\mathrm{in}\ F.
\end{equation}
Milner's approach has been applied to branching congruence by van Glabbeek~\cite{vanGlabbeek1993}.
It has been extended and applied to the value-passing calculi by Hennessy, Lin and
Rathke~\cite{HennessyLin1997,Rathke1997,HennessyLinRathke1997} and
to the $\pi$-calculus by Lin~\cite{Lin1995-fix-ind,Lin1998}. All
these more complicated complete proof systems are of a symbolic
nature. The side condition of the fixpoint induction
(\ref{fixinduction}) renders the rule truly unwelcome. But as Sewell
has proved in~\cite{Sewell1994,Sewell1997} there is no finitely
axiomatizable complete system for the finite controls. In order to
derive the equality
\[\mu X.a.X=\mu X.\underset{k>1}{\underbrace{a.\cdots.a}}.X\]
from a finite system, one needs rule(s) in addition to axioms. It is
possible to come up with a complete proof system in which all the
rules are unconditional~\cite{Sewell1995}. But it would probably not
pay off when it comes down to implementation. Now if we have to
stick to the fixpoint induction, how should we make use of it?
Milner's straightforward answer~\cite{Milner1984}, adopted in almost
all follow-up work, is to introduce process equation systems.
To apply this approach to the $\pi$-calculus, it is convenient to use abstractions over names~\cite{HennessyLinRathke1997,Lin1998}.

The finite states/controls raise the question of divergence. The
axiomatic treatment of divergence has been borrowing ideas from
domain theory~\cite{AmadioCurien1998}. Scott's denotational approach
is too abstract to give a proper account of interactions, and in the
case of divergence, non-interactions. Early treatment of divergence in process algebra is more influenced by the domain theoretical approach~\cite{DeNicolaHennessy1984,Walker1990}. It appears that the first successful operational approach to divergence is achieved
in Lohrey, D'Argenio and Hermanns'
work~\cite{LohreyDArgenioHermanns2002,LohreyDArgenioHermanns2005}
on the axioms for divergence. Crucial
to their approach is the observation that all divergence of a finite
control is due to self-looping. The $\Delta$-operator, defined
below, is isolated to play a key role in their inference systems.
\begin{equation}\label{DeltaOp}
\Delta(T) \stackrel{\rm def}{=} \mu X.(\tau.X+T)
\end{equation}
Lohrey, D'Argenio and Hermanns' systems consist of the laws to
convert divergence from one form to another so that the fixpoint
induction can be applied without any regard to divergence. One
axiom proposed in~\cite{LohreyDArgenioHermanns2002,LohreyDArgenioHermanns2005} is
\begin{equation}\label{LohreyDArgenioHermannsweak}
\Delta(\Delta(T)+T') = \tau.(\Delta(T)+T').
\end{equation}
The law (\ref{LohreyDArgenioHermannsweak}) is sound for the
termination preserving weak congruence. But it fails to meet the
codivergence property. Based on Lohrey, D'Argenio and Hermanns'
work, Fu discusses the axioms for the codivergence in the framework of CCS~\cite{FuYuxi-Nondetrministic-Computation}.
The codivergent version of (\ref{LohreyDArgenioHermannsweak}) for example is
\begin{equation}\label{FHweak}
\Delta(\Delta(T)) = \Delta(T).
\end{equation}
It is worth remarking that (\ref{FHweak}) is valid for the strong equality.
The proof systems using the $\Delta$-operator begin to unveil the rich structure of divergence.
It also provides a purely equational characterization of the internal actions.
The codivergence satisfies another equational axiom:
\begin{equation}\label{2011-02-12}
\Delta(S{+}\tau.\Delta(S{+}S')) =\Delta(S{+}S').
\end{equation}
It is interesting to notice the similarity between (\ref{2011-02-12}) and (\ref{2010-01-14}).

The above discussion is meant to bring out the following point.
The system $AS$ plus the axioms of codivergence proposed in~\cite{FuYuxi-Nondetrministic-Computation} give rise to a complete proof system for the absolute equality on the regular $\pi$-processes, which is more accessible than the symbolic proof system.

How about a complete equational system for the box equality on the regular $\pi$-processes.
A crucial step would be to prove for the regular $\pi$-processes a property corresponding to the one stated in Lemma~\ref{box2absolute}.
It could turn out that this problem is far more difficult than one would have perceived.
Rensink and Vogler~\cite{RensinkVogler2007} point out a surprising fact that Milner's fixpoint induction fails for the fair/should testing equivalence.
Consider the process equation
\begin{eqnarray}
X &=& a.X+a.\mu Z.(a.Z+b). \label{2010-06-16}
\end{eqnarray}
It is not difficult to see that
\[a.\mu Z.a.Z+a.\mu Z.(a.Z+b) =_{\Box} a.(a.\mu Z.a.Z+a.\mu Z.(a.Z+b))+a.\mu Z.(a.Z+b).\]
The solution
\begin{eqnarray}\label{RV-1}
RV_{0} &\stackrel{\rm def}{=}& a.\mu Z.a.Z+a.\mu Z.(a.Z+b)
\end{eqnarray}
is not box equal to the canonical solution
\begin{eqnarray}\label{RV-2}
RV_{1} &\stackrel{\rm def}{=}& \mu X.(a.X+a.\mu Z.(a.Z+b))
\end{eqnarray}
to the equation (\ref{2010-06-16}).
A context that tells apart the processes (\ref{RV-1}) and (\ref{RV-2}) is $(ab)(\_\,|\,\mu Y.(\overline{a}.Y+\overline{b}.d))$.
This is because $(ab)(RV_{1}\,|\,\mu Y.(\overline{a}.Y+\overline{b}.d))$ is strongly observable whereas $(ab)(RV_{0}\,|\,\mu Y.(\overline{a}.Y+\overline{b}.d))$ is not.

Unlike the pure algebraic view of process~\cite{BaetenWeijland1990}, we tend to think of $AS$ and $AS_{b}$ as proof systems that help derive process equalities. This explains the particular statements of Theorem~\ref{ae:completeness} and of Theorem~\ref{box:completeness}. The emphasis on equational proof systems rather than on axiomatic systems has the advantage that the introduction of congruence relations can be avoided. What is stated in Theorem~\ref{ae:completeness} is called promotion property in~\cite{FuYang2003}. It is pointed out by Fu and Yang~\cite{FuYang2003} that this property is absolutely necessary to prove that an algebraic system of the $\pi$-calculus is complete, the reason being that Hennessy Lemma~\cite{Milner1989} fails for the $\pi$-calculus~\cite{FuYang2003}.

\section{Future Work}\label{sec-future-work}

The studies in process calculi over the last thirty years largely fall into two main categories:
\begin{enumerate}
\item [I.] The first is about the diversity of models.
A lot of process calculi have been designed and discussed~\cite{Nestmann2006}.

\item [II.] The second is about equivalence.
Many observational equivalences and algebraic equalities have been proposed, axiomatized and compared~\cite{vanGlabbeek1993II,vanGlabbeek2001I}.
\end{enumerate}
In literature there is only a small number of papers that deal with expressiveness or completeness issues~\cite{BusiGabbrielliZavattaro2003,BusiGabbrielliZavattaro2004,Palamidessi2003,Gorla2009MFPS,FuLu2010}.
If we understand the situation correctly, we are at a stage where we try to seek the {\em eternal truth} in process theory~\cite{Abramsky2006}.
It is our opinion that we will not be able to go very far if we do not seriously address the issue of {\em problem solving} with process calculi.
It is only by applying a model to tackle real problems can our treatment of the observational theory of the model be verified.
Talking about problem solving, there is no better process calculus than the $\pi$-calculus to start with such an investigation.
What is achieved in this paper is a condensed account of the foundational theory of the $\pi$-models that provides support for problem solving.

Eternal truths have to be model independent.
They would not be very interesting if the models we are considering are too liberal and too diversified.
This is why our presentation of the $\pi$-calculus has followed the general principles and methodologies of Theory of Interaction developed in~\cite{FuYuxi}.
The prime motivation of Theory of Interaction can be summarized as follows:
\begin{quote}
There are two eternal relationships in computer science, one is given by the equality relation (absolute equality) between the programs (or processes) of a model, and the other is by the expressiveness relation (subbisimilarity) between the models.
By using these two fundamental relations, one can then introduce a number of basic postulates that formalize the foundational assumptions widely adopted in computer science.
Now let $\mathfrak{M}$ be the class of all models.
The first postulate asserts that a model belonging to $\mathfrak{M}$ must be computationally complete.
\vspace*{2mm}
\begin{quote}
{\bf Axiom of Completeness}.
$\forall\;\mathbb{M}\in\mathfrak{M}.\;\mathbb{C}\sqsubseteq\mathbb{M}$.
\end{quote}
\vspace*{2mm}
The Computability Model $\mathbb{C}$, defined in~\cite{FuYuxi}, is the minimal interactive extension of the computable function model.
The Axiom of Completeness can be seen as a formalization of Church-Turing Thesis.
It places a considerable constraint on the world  $\mathfrak{M}$ of models.
\end{quote}
In~\cite{FuYuxi} it is shown that $\pi^{M}$ is complete in the sense that it satisfies the Axiom of Completeness.
In~\cite{XueLongFu2011} the completeness of $\pi^{L}$, $\pi^{R}$ and $\pi^{S}$ is established.
So it makes sense to talk about problem solving in all these $\pi$-variants.
We remark that completeness in our sense is much stronger than the so-called Turing completeness~\cite{FuLu2010}.
Some variants of CCS~\cite{Milner1989} are Turing complete~\cite{BusiGabbrielliZavattaro2003,BusiGabbrielliZavattaro2003}, they are however not complete in our stronger sense~\cite{FuYuxi}.

What is done in this paper can be seen as a book-keeping exercise.
New results are obtained and old results are assessed in a uniform formalism.
Based on what is set up in this paper, the $\pi$-model can be studied in three main directions.
\begin{enumerate}
\item [III.] A theory of $\pi$-solvability that goes beyond the traditional
recursion theory should be useful to understand the power of the
name-passing calculi. It can be shown that a pseudo natural number
generator is solvable in the $\pi$-calculus. But a genuine natural number
generator is $\pi$-unsolvable. It would be useful to develop a
theory of $\pi$-solvability and investigate the diagonal methods for
tackling $\pi$-unsolvability. Other possible issues to look at are
nondeterministic functions definable in $\pi$, the recursion theory of the $\pi$-processes (say the formulation of the s-m-n theorem, enumeration theorem, recursion theorem etc.).

\item [IV.] A comparative study of the complexity classes defined by the $\pi$-programs
against the standard complexity classes~\cite{Papadimitriou1994}
would be instructive for a better understanding of the algorithmic aspect of the $\pi$-calculus.
Such a study may begin with a formulation of the class $\textsf{P}_{\pi}$ of the
problems decidable in the $\pi$-calculus in polynomial time in a way the class $\textsf{NP}$ is defined in terms of the Nondeterministic Turing Machine
Model, and then investigate the relationship between  $\textsf{P}_{\pi}$ and $\textsf{NP}$.

\item [V.] At a more applied level, one could try to develop a hierarchy of programming languages implemented on the $\pi$-calculus.
Previous studies in the program theory of the $\pi$-calculus have not discussed anything about universal process for $\pi$.
So a great deal more has to be learned before we are confident of using $\pi$ as a machine model to implement a full fledged typed programming language.
\end{enumerate}
One may question the practical relevance of these new research directions.
After all the $\pi$-calculus, unlike the Deterministic Turing Machine Model, does not have a physical implementation.
The rapid development of computing technology has actually provided an answer.
The Internet is an approximate implementation of the $\pi$-calculus!
It is not completely, as one might argue, a physical implementation.
But it is the kind of computing environment for which a theory of the $\pi$-calculus might provide just the right foundation.


\newpage

\large\noindent {\bf Acknowledgments}\normalsize

\vspace*{2mm}

This work has been supported by the National Natural Science Foundation of China (grant numbers 60873034, 61033002).
The authors would like to thank the members of BASICS for their interest and feedbacks.
Especially they would like to thank Xiaojuan Cai and Chaodong He for proof reading the paper.
Chaodong He has pointed out to us that our previous proof of Lemma~\ref{lem:promotion} and the previous statement of Theorem~\ref{2008-09-24} were mistaken.
He actually showed us how to correct the mistakes.
The authors would also like to thank Huan Long and Jianxin Xue for their discussions on the previous versions of the paper.
 Their comments have led to several improvements of the paper.

The first author would like to thank Huan Long and Jianxin Xue for sharing the interest in the $\pi$-variants $\pi^{L},\pi^{R},\pi^{S}$.
Some of the problems raised in a previous version of the paper are resolved in~\cite{XueLongFu2011}.




\begin{thebibliography}{100}

\bibitem{Abramsky2006}
S.~Abramsky.
\newblock {What are the Fundamental Structures of Concurrency? We still do not
  know.}
\newblock Electronic Notes in Theoretical Computer Science, pages 37--41, 2006.

\bibitem{AmadioCastellaniSangiorgi1996}
R.~Amadio, I.~Castellani, and D.~Sangiorgi.
\newblock On bisimulations for the asynchronous $\pi$-calculus.
\newblock In {\em Proc. CONCUR'96}, volume 1119 of {\em Lecture Notes in
  Computer Science}, pages 147--162. Springer, 1996.

\bibitem{AmadioCurien1998}
R.~Amadio and P.~Curien.
\newblock {\em Domains and Lambda-Calculi}.
\newblock Cambridge Tracts in Theoretical Computer Science. Cambridge
  University Press, 1998.

\bibitem{Baeten1996}
J.~Baeten.
\newblock Branching bisimilarity is an equivalence indeed.
\newblock {\em Information Processing Letters}, 58:141--147, 1996.

\bibitem{BaetenWeijland1990}
J.~Baeten and W.~Weijland.
\newblock {\em Process Algebra}, volume~18 of {\em Cambridge Tracts in
  Theoretical Computer Science}.
\newblock CUP, 1990.

\bibitem{BaierHermanns1997}
C.~Baier and H.~Hermanns.
\newblock Weak bisimulation for fully probabilistic processes.
\newblock In {\em Proc. CAV'97}, volume 1254 of {\em Lecture Notes in Computer
  Science}, pages 119--130, 1997.

\bibitem{BergstraKlopOlderog1987}
J.~Bergstra, J.~Klop, and E.~Olderog.
\newblock Failures without chaos: A new process semantics for fair abstraction.
\newblock In {\em Formal Descprition of Programming Concepts -- III}, IFIP,
  pages 77--103. Elsevier Science Publisher, 1987.

\bibitem{Boreale1996}
M.~Boreale.
\newblock On the expressiveness of internal mobility in name-passing calculi.
\newblock In {\em Proc. CONCUR'96}, volume 1119 of {\em Lecture Notes in
  Computer Science}, pages 161--178, 1996.

\bibitem{BorealeDeNicola1995}
M.~Boreale and R.~De~Nicola.
\newblock Testing equivalence for mobile processes.
\newblock {\em Information and Computation}, 120:279--303, 1995.

\bibitem{BorealeDeNicolaPugliese1999}
M.~Boreale, R.~De~Nicola, and R.~Pugliese.
\newblock Basic observables for processes.
\newblock {\em Information and Computation}, 149:77--98, 1999.

\bibitem{BorealeDeNicolaPugliese2001}
M.~Boreale, R.~De~Nicola, and R.~Pugliese.
\newblock Divergence in testing and readiness semantics.
\newblock {\em Theoretical Computer Science}, 266:237--248, 2001.

\bibitem{Boudol1992}
G.~Boudol.
\newblock Asynchrony and the $\pi$-calculus.
\newblock Technical Report RR-1702, INRIA Sophia-Antipolis, 1992.

\bibitem{BrinksmaRensinkVogler1995}
E.~Brinksma, A.~Rensink, and W.~Vogler.
\newblock Fair testing.
\newblock In {\em Proc. CONCUR'95}, volume 962 of {\em Lecture Notes in
  Computer Science}, pages 313--327, 1995.

\bibitem{BusiGabbrielliZavattaro2003}
N.~Busi, M.~Gabbrielli, and G.~Zavattaro.
\newblock Replication vs recursive definitions in channel based calculi.
\newblock In {\em Proc. ICALP'03}, volume 2719 of {\em Lecture Notes in
  Computer Science}, pages 133--144, 2003.

\bibitem{BusiGabbrielliZavattaro2004}
N.~Busi, M.~Gabbrielli, and G.~Zavattaro.
\newblock Comparing recursion, replication and iteration in process calculi.
\newblock In {\em Proc. ICALP'04}, volume 3142 of {\em Lecture Notes in
  Computer Science}, pages 307--319, 2004.

\bibitem{CacciagranoCorradiniArandaValencia2008}
D.~Cacciagrano, F.~Corradini, J.~Aranda, and F.~Valencia.
\newblock Linearity, persistence and testing semantics in the asynchronous
  pi-calculus.
\newblock {\em Electronic Notes in Theoretical Computer Science}, 194:59--84,
  2008.

\bibitem{CacciagranoCorradiniPalamidessi2006}
D.~Cacciagrano, F.~Corradini, and C.~Palamidessi.
\newblock Separation of synchronous and asynchronous communication via testing.
\newblock {\em Electronic Notes in Theoretical Computer Science}, 154:95--108,
  2006.

\bibitem{CaiFu2011}
X.~Cai and Y.~Fu.
\newblock The $\lambda$-calculus in the $\pi$-calculus.
\newblock {\em Mathematical Structure in Computer Science}, 21:943--996, 2011.

\bibitem{DeNicolaHennessy1984}
R.~De~Nicola and M.~Hennessy.
\newblock Testing equivalence for processes.
\newblock {\em Theoretical Computer Science}, 34:83--133, 1984.

\bibitem{DeNicolaMontanariVaandrager1990}
R.~De~Nicola, U.~Mantanari, and F.~Vaandrager.
\newblock Back and forth bisimulations.
\newblock In {\em Proc. CONCUR'90}, volume 458 of {\em Lecture Notes in
  Computer Science}, pages 152--165, 1990.

\bibitem{DeNicolaVaandrager1995}
R.~De~Nicola and F.~Vaandrager.
\newblock Three logics for branching bisimulation.
\newblock {\em Journal of ACM}, 42:458--487, 1995.

\bibitem{FournetGonthier1996}
C.~Fournet and G.~Gonthier.
\newblock The reflexive chemical abstract machines and the join calculus.
\newblock In {\em Proc. POPL'96}. ACM Press, 1996.

\bibitem{FuYuxi1999}
Y.~Fu.
\newblock Variations on mobile processes.
\newblock {\em Theoretical Computer Science}, 221:327--368, 1999.

\bibitem{FuYuxi2005}
Y.~Fu.
\newblock On quasi open bisimulation.
\newblock {\em Theoretical Computer Science}, 338:96--126, 2005.

\bibitem{FuYuxi-Theory-by-Process}
Y.~Fu.
\newblock Theory by process.
\newblock In {\em CONCUR 2010}, Lecture Notes in Computer Science, Paris,
  France, 2010.

\bibitem{FuYuxi-Nondetrministic-Computation}
Y.~Fu.
\newblock Nondeterministic structure of computation.
\newblock {\em To appear in Mathematical Structures in Computer Science}, 2015.

\bibitem{FuYuxi}
Y.~Fu.
\newblock Theory of interaction.
\newblock {\em Theoretical Computer Science}, accepted, 2015.

\bibitem{FuLu2010}
Y.~Fu and H.~Lu.
\newblock On the expressiveness of interaction.
\newblock {\em Theoretical Computer Science}, 411:1387--1451, 2010.

\bibitem{FuYang2003}
Y.~Fu and Z.~Yang.
\newblock Tau laws for pi calculus.
\newblock {\em Theoretical Computer Science}, 308:55--130, 2003.

\bibitem{Gorla2009MFPS}
D.~Gorla.
\newblock On the relative power of calculi for mobility.
\newblock In {\em Proc. MFPS'09}, volume 249 of {\em Electronic Notes in
  Theoretical Computer Science}, pages 269--286, 2009.

\bibitem{He2010}
C.~He.
\newblock Model independent order relations for processes.
\newblock In {\em APLAS 2010}, volume 6461 of {\em Lecture Notes in Computer
  Science}, pages 408--423, 2010.

\bibitem{Hennessy1988}
M.~Hennessy.
\newblock {\em An Algebraic Theory of Processes}.
\newblock MIT Press, Cambridge, MA, 1988.

\bibitem{Hennessy1991}
M.~Hennessy.
\newblock A proof system for communicating processes with value-passing.
\newblock {\em Journal of Formal Aspects of Computer Science}, 3:346--366,
  1991.

\bibitem{HennessyIng1993a}
M.~Hennessy and A.~Ing\'{o}lfsd\'{o}ttir.
\newblock Communicating processes with value-passing and assignment.
\newblock {\em Journal of Formal Aspects of Computing}, 5:432--466, 1993.

\bibitem{HennessyIng1993b}
M.~Hennessy and A.~Ing\'{o}lfsd\'{o}ttir.
\newblock A theory of communicating processes with value-passing.
\newblock {\em Information and Computation}, 107:202--236, 1993.

\bibitem{HennessyLin1995}
M.~Hennessy and H.~Lin.
\newblock Symbolic bisimulations.
\newblock {\em Theoretical Computer Science}, 138:353--369, 1995.

\bibitem{HennessyLin1996}
M.~Hennessy and H.~Lin.
\newblock Proof systems for message passing process algebras.
\newblock {\em Formal Aspects of Computing}, 8:379--407, 1996.

\bibitem{HennessyLin1997}
M.~Hennessy and H.~Lin.
\newblock Unique fixpoint induction for message-passing process calculi.
\newblock In {\em Proc. Computing: Australian Theory Symposium (CAT'97)},
  volume~8, pages 122--131, 1997.

\bibitem{HennessyLinRathke1997}
M.~Hennessy, H.~Lin, and J.~Rathke.
\newblock Unique fixpoint induction for message-passing process calculi.
\newblock {\em Science of Computer Programming}, 41:241--275, 1997.

\bibitem{HennessyMilner1985}
M.~Hennessy and R.~Milner.
\newblock Algebraic laws for nondeterminism and concurrency.
\newblock {\em Journal of ACM}, 32:137--161, 1985.

\bibitem{HondaTokoro1991ObjCal}
K.~Honda and M.~Tokoro.
\newblock An object calculus for asynchronous communications.
\newblock In {\em Proc. ECOOP'91}, volume 512 of {\em Lecture Notes in Computer
  Science}, pages 133--147, Geneva, Switzerland, 1991.

\bibitem{HondaTokoro1991AsynSem}
K.~Honda and M.~Tokoro.
\newblock On asynchronous communication semantics.
\newblock In {\em Proc. Workshop on Object-Based Concurrent Computing}, volume
  615 of {\em Lecture Notes in Computer Science}, pages 21--51, 1991.

\bibitem{HondaYoshida1995}
K.~Honda and M.~Yoshida.
\newblock On reduction-based process semantics.
\newblock {\em Theoretical Computer Science}, 151:437--486, 1995.

\bibitem{IngolfsdottirLin2001}
A.~Ing\'{o}lfsd\'{o}ttir and H.~Lin.
\newblock A symbolic approach to value-passing processes.
\newblock In J.~Bergstra, A.~Ponse, and S.~Smolka, editors, {\em Handbook of
  Process Algebra}, pages 427--478. North-Holland, 2001.

\bibitem{Lin1995}
H.~Lin.
\newblock Complete inference systems for weak bisimulation equivalences in the
  $\pi$-calculus.
\newblock In {\em Proceedings of Sixth International Joint Conference on the
  Theory and Practice of Software Development}, volume 915 of {\em Lecture
  Notes in Computer Science}, pages 187--201, 1995.

\bibitem{Lin1995-fix-ind}
H.~Lin.
\newblock Unique fixpoint induction for mobile processes.
\newblock In {\em Proc. CONCUR '95}, volume 962 of {\em Lecture Notes in
  Computer Science}, pages 88--102, 1995.

\bibitem{Lin1996}
H.~Lin.
\newblock Symbolic transition graphs with assignment.
\newblock In {\em Proc. CONCUR '96}, volume 1119 of {\em Lecture Notes in
  Computer Science}, pages 50--65, 1996.

\bibitem{Lin1998}
H.~Lin.
\newblock Complete proof systems for observation congruences in finite-control
  $\pi$-calculus.
\newblock In {\em Proc. ICALP '98}, volume 1443 of {\em Lecture Notes in
  Computer Science}, pages 443--454, 1998.

\bibitem{Lin2003}
H.~Lin.
\newblock Complete inference systems for weak bisimulation equivalences in the
  pi-calculus.
\newblock {\em Information and Computation}, 180:1--29, 2003.

\bibitem{LohreyDArgenioHermanns2002}
M.~Lohrey, P.~D'Argenio, and H.~Hermanns.
\newblock Axiomatising divergence.
\newblock In {\em Proc. ICALP 2002}, volume 2380 of {\em Lecture Notes in
  Computer Science}, pages 585--596. Springer, 2002.

\bibitem{LohreyDArgenioHermanns2005}
M.~Lohrey, P.~D'Argenio, and H.~Hermanns.
\newblock Axiomatising divergence.
\newblock {\em Information and Computatio}, 203:115--144, 2005.

\bibitem{Merro2000}
M.~Merro.
\newblock {\em Locality in the $\pi$-Calculus and Applications to
  Object-Oriented Languages}.
\newblock PhD thesis, Ecole des Mines de Paris, 2000.

\bibitem{MerroSangiorgi2004}
M.~Merro and D.~Sangiorgi.
\newblock On asynchrony in name-passing calculi.
\newblock {\em Mathematical Structures in Computer Science}, 14:715--767, 2004.

\bibitem{Milner1980}
R.~Milner.
\newblock A calculus of communicating systems.
\newblock {\em Lecture Notes in Computer Science}, 92, 1980.

\bibitem{Milner1984}
R.~Milner.
\newblock A complete inference system for a class of regular behaviours.
\newblock {\em Journal of Computer and System Science}, 28:439--466, 1984.

\bibitem{Milner1989}
R.~Milner.
\newblock {\em Communication and Concurrency}.
\newblock Prentice Hall, 1989.

\bibitem{Milner1989IC}
R.~Milner.
\newblock A complete axiomatization system for observational congruence of
  finite state behaviours.
\newblock {\em Information and Computation}, 81:227--247, 1989.

\bibitem{Milner1993}
R.~Milner.
\newblock Elements of interaction.
\newblock {\em Communication of ACM}, 36:78--89, 1993.

\bibitem{Milner1993-poly}
R.~Milner.
\newblock The polyadic $\pi$-calculus: a tutorial.
\newblock In {\em Proceedings of the 1991 Marktoberdorf Summer School on Logic
  and Algebra of Specification}, NATO ASI, Series F. Springer-Verlag, 1993.

\bibitem{MilnerParrowWalker1992}
R.~Milner, J.~Parrow, and D.~Walker.
\newblock A calculus of mobile processes.
\newblock {\em Information and Computation}, 100:1--40 (Part I), 41--77 (Part
  II), 1992.

\bibitem{MilnerSangiorgi1992}
R.~Milner and D.~Sangiorgi.
\newblock Barbed bisimulation.
\newblock In {\em Proc. ICALP'92}, volume 623 of {\em Lecture Notes in Computer
  Science}, pages 685--695, 1992.

\bibitem{NatarajanCleaveland1995}
V.~Natarajan and R.~Cleaveland.
\newblock Divergence and fair testing.
\newblock In {\em Proc. ICALP'95}, volume 944 of {\em Lecture Notes in Computer
  Science}, pages 648--659, 1995.

\bibitem{Nestmann2000}
U.~Nestmann.
\newblock What is a good encoding of guarded choices?
\newblock {\em Information and computation}, 156:287--319, 2000.

\bibitem{Nestmann2006}
U.~Nestmann.
\newblock Welcome to the jungle: A subjective guide to mobile process calculi.
\newblock In {\em Proc. CONCUR'06}, volume 4137 of {\em Lecture Notes in
  Computer Science}, pages 52--63, 2006.

\bibitem{NestmannPierce1996}
U.~Nestmann and B.~Pierce.
\newblock Decoding choice encodings.
\newblock In U.~Montanari and V.~Sassone, editors, {\em Proc. CONCUR'96},
  volume 1119 of {\em Lecture Notes in Computer Science}, pages 179--194, 1996.

\bibitem{PalamidessiSaraswatValenciaVictor2006}
C.~Palamiddessi, V.~Saraswat, F.~Valencia, and B.~Victor.
\newblock On the expressiveness of linearity vs. persistence in the
  asynchronous pi calculus.
\newblock In {\em Proc. LICS'06}, pages 59--68. IEEE press, 2006.

\bibitem{Palamidessi2003}
C.~Palamidessi.
\newblock Comparing the expressive power of the synchronous and the
  asynchronous $\pi$-calculus.
\newblock {\em Mathematical Structures in Computer Science}, 13:685--719, 2003.

\bibitem{Papadimitriou1994}
C.~Papadimitriou.
\newblock {\em Computational Complexity}.
\newblock Addison-Wesley, Reading, MA, 1994.

\bibitem{Park1981}
D.~Park.
\newblock Concurrency and automata on infinite sequences.
\newblock In {\em Theoretical Computer Science}, volume Lecture Notes in
  Computer Science 104, pages 167--183. Springer, 1981.

\bibitem{Parrow1999}
J.~Parrow.
\newblock On the relationship between two proof systems for the pi-calculus,
  1999.

\bibitem{Parrow2001}
J.~Parrow.
\newblock An introduction to the $\pi$-calculus.
\newblock In J.~Bergstra, A.~Ponse, and S.~Smolka, editors, {\em Handbook of
  Process Algebra}, pages 478--543. North-Holland, 2001.

\bibitem{ParrowSangiorgi1995}
J.~Parrow and D.~Sangiorgi.
\newblock Algebraic theories for name-passing calculi.
\newblock {\em Information and Computation}, 120:174--197, 1995.

\bibitem{Phillips1987}
I.~Phillips.
\newblock Refusal testing.
\newblock {\em Theoretical Computer Science}, 50:241--284, 1987.

\bibitem{PierceTurner2000}
B.~Pierce and D.~Turner.
\newblock Pict: A programming language based on the pi calculus.
\newblock In G.~Plotkin, C.~Stirling, and M.~Tofts, editors, {\em Proof,
  Language and Interaction: Essays in Honour of Robin Milner}. MIT Press, 2000.

\bibitem{Priese1978}
L.~Priese.
\newblock On the concept of simulation in asynchronous, concurrent systems.
\newblock {\em Progress in Cybernatics and Systems Research}, 7:85--92, 1978.

\bibitem{Rathke1997}
J.~Rathke.
\newblock Unique fixpoint induction for value-passing processes.
\newblock In {\em Proc. LICS '97}, IEEE Press, 1997.

\bibitem{RensinkVogler2007}
A.~Rensink and W.~Vogler.
\newblock Fair testing.
\newblock {\em Information and Computation}, 205:125--198, 2007.

\bibitem{Rogers1987}
H.~Rogers.
\newblock {\em Theory of Recursive Functions and Effective Computability}.
\newblock MIT Press, 1987.

\bibitem{Sangiorgi1992}
D.~Sangiorgi.
\newblock {\em Expressing Mobility in Process Algebras: First Order and Higher
  Order Paradigm}.
\newblock PhD thesis, Department of Computer Science, University of Edinburgh,
  1992.

\bibitem{Sangiorgi1993}
D.~Sangiorgi.
\newblock From $\pi$-calculus to higher order $\pi$-calculus--and back.
\newblock In {\em Proc. TAPSOFT'93}, volume 668 of {\em Lecture Notes in
  Computer Science}, pages 151--166, 1993.

\bibitem{Sangiorgi1996IC}
D.~Sangiorgi.
\newblock Bisimulation for higher order process calculi.
\newblock {\em Information and Comptation}, 131:141--178, 1996.

\bibitem{Sangiorgi1996TCS}
D.~Sangiorgi.
\newblock $\pi$-calculus, internal mobility and agent-passing calculi.
\newblock {\em Theoretical Computer Science}, 167:235--274, 1996.

\bibitem{Sangiorgi1996AI}
D.~Sangiorgi.
\newblock A theory of bisimulation for $\pi$-calculus.
\newblock {\em Acta Informatica}, 3:69--97, 1996.

\bibitem{Sangiorgi2009}
D.~Sangiorgi.
\newblock On the origin of bisimulation and coinduction.
\newblock {\em Transactions on Programming Languages and Systems}, 31(4), 2009.

\bibitem{SangiorgiMilner1992}
D.~Sangiorgi and R.~Milner.
\newblock Techniques of ``weak bisimulation up to''.
\newblock In {\em Proc. CONCUR'92}, volume 630 of {\em Lecture Notes in
  Computer Science}, pages 32--46, 1992.

\bibitem{SangiorgiWalker2001}
D.~Sangiorgi and D.~Walker.
\newblock On barbed equivalence in $\pi$-calculus.
\newblock In {\em Proc. CONCUR'01}, volume 2154 of {\em Lecture Notes in
  Computer Science}, pages 292--304, 2001.

\bibitem{SangiorgiWalker2001BOOK}
D.~Sangiorgi and D.~Walker.
\newblock {\em The $\pi$ Calculus: A Theory of Mobile Processes}.
\newblock Cambridge University Press, 2001.

\bibitem{Sewell1994}
P.~Sewell.
\newblock Bisimulation is not finitely (first order) equationally
  axiomatisable.
\newblock In {\em Proc. LICS'94}, IEEE, pages 62--70, 1994.

\bibitem{Sewell1995}
P.~Sewell.
\newblock {\em The Algebra of Finite State Processes}.
\newblock PhD thesis, The University of Edinburgh, 1995.

\bibitem{Sewell1997}
P.~Sewell.
\newblock Nonaxiomatisability of equivalence over finite state processes.
\newblock {\em Annals of Pure and Applied Logic}, 90:163--191, 1997.

\bibitem{Thomsen1989}
B.~Thomsen.
\newblock A calculus of higher order communicating systems.
\newblock In {\em Proc. POPL'89}, pages 143--154, 1989.

\bibitem{Thomsen1990}
B.~Thomsen.
\newblock {\em Calculi for Higher Order Communicating Systems}.
\newblock PhD thesis, Department of Computing, University of London, 1990.

\bibitem{Thomsen1993}
B.~Thomsen.
\newblock Plain chocs --- a second generation calculus for higher order
  processes.
\newblock {\em Acta Informatica}, 30:1--59, 1993.

\bibitem{Thomsen1995}
B.~Thomsen.
\newblock A theory of higher order communicating systems.
\newblock {\em Information and Computation}, 116:38--57, 1995.

\bibitem{vanGlabbeek1990}
R.~van Glabbeek.
\newblock Linear time -- branching time spectrum.
\newblock In {\em Proc. CONCUR'90}, volume 458 of {\em Lecture Notes in
  Computer Science}, pages 278--297, 1990.

\bibitem{vanGlabbeek1993}
R.~van Glabbeek.
\newblock A complete axiomatization for branching bisimulation congruence of
  finite-state behaviours.
\newblock In {\em Proc. MFCS'93}, volume 711 of {\em Lecture Notes in Computer
  Science}, pages 473--484, 1993.

\bibitem{vanGlabbeek1993II}
R.~van Glabbeek.
\newblock Linear time -- branching time spectrum (ii).
\newblock In {\em Proc. CONCUR'93}, volume 715 of {\em Lecture Notes in
  Computer Science}, pages 66--81, 1993.

\bibitem{vanGlabbeek1994}
R.~van Glabbeek.
\newblock What is branching time semantics and why to use it?
\newblock In G.~Paun, G.~Rozenberg, and A.~Salomaa, editors, {\em Current
  Trends in Theoretical Computer Science; Entering the 21th Century}, World
  Scientific, pages 469--479, 1994.

\bibitem{vanGlabbeek2001I}
R.~van Glabbeek.
\newblock Linear time -- branching time spectrum (i).
\newblock In J.~Bergstra, A.~Ponse, and S.~Smolka, editors, {\em Handbook of
  Process Algebra}, pages 3--99. North-Holland, 2001.

\bibitem{vanGlabbeekLuttikTrcka2009}
R.~van Glabbeek, B.~Luttik, and N.~Tr\v{c}ka.
\newblock Branching bisimilarity with explicit divergence.
\newblock {\em Fundamenta Informaticae}, 93:371--392, 2009.

\bibitem{vanGlabbeekWeijland1989-first-paper-bb}
R.~van Glabbeek and W.~Weijland.
\newblock Branching time and abstraction in bisimulation semantics.
\newblock In {\em Information Processing'89}, pages 613--618. North-Holland,
  1989.

\bibitem{Walker1990}
D.~Walker.
\newblock Bisimulation and divergence.
\newblock {\em Information and Computation}, 85:202--241, 1990.

\bibitem{Walker1991}
D.~Walker.
\newblock $\pi$-calculus semantics for object-oriented programming languages.
\newblock In {\em Proc. TACS '91}, volume 526 of {\em Lecture Notes in Computer
  Science}, pages 532--547, 1991.

\bibitem{Walker1995}
D.~Walker.
\newblock Objects in the $\pi$-calculus.
\newblock {\em Information and Computation}, 116:253--271, 1995.

\bibitem{XueLongFu2011}
J.~Xue, H.~Long, and Y.~Fu.
\newblock Remark on some pi variants.
\newblock 2011.

\end{thebibliography}
\end{document}